\documentclass[lettersize,journal]{IEEEtran} 

\IEEEoverridecommandlockouts                              
\usepackage{color} 
\usepackage{bm}

\usepackage{amsmath,amsthm} 
\usepackage{graphics} 
\usepackage{epsfig} 
\usepackage{times} 
\usepackage{amssymb}  
\usepackage{mathtools}
\usepackage{caption}
\usepackage{subcaption}
\usepackage{cite}
\usepackage{wrapfig}
\theoremstyle{definition}
\newtheorem{problem}{Problem}
\newtheorem{example}{Example}
\newtheorem{assumption}{Assumption}
\newtheorem{definition}{Definition}

\newtheorem{remark}{Remark}
\newtheorem{theorem}{Theorem}
\newtheorem{corollary}{Corollary}
\newtheorem{case study}{Case Study}
\usepackage[linesnumbered,ruled,vlined]{algorithm2e}
\usepackage{comment}
\title{\LARGE \bf
Feasibility-aware Learning of Robust Temporal Logic Controllers using BarrierNet}

\author{Wenliang Liu$^{1*}$, Shuo Liu$^{2*}$, Wei Xiao$^{3}$, and Calin A. Belta$^{4}$
\thanks{This work was partially supported by the National Science Foundation under grant IIS-2024606 at Boston University.}
\thanks{$^{1}$W. Liu is with Amazon Robotics, North Reading, MA USA.
        {\tt\small liuwll@amazon.com}}%
\thanks{$^{2}$S. Liu is with Department of Mechanical Engineering,
        Boston University, MA, USA.
        {\tt\small liushuo@bu.edu}}%
       
\thanks{$^{3}$W. Xiao is with Department of  Robotics Engineering, Worcester Polytechnic Institute, MA, USA.
        {\tt\small wxiao3@wpi.edu}}%
\thanks{$^{4}$C. Belta is with the Department of Electrical and Computer Engineering and the Department of Computer Science, University of Maryland, College Park, MD, USA. 
        {\tt\small cbelta@umd.edu}}%
\thanks{$^{*}$These author contributed equally. }%

}

\begin{document}

\maketitle
\thispagestyle{empty}
\pagestyle{empty}

\begin{abstract}
Control Barrier Functions (CBFs) have been used to enforce safety and task specifications expressed in Signal Temporal Logic (STL). However, existing CBF-STL approaches typically rely on fixed hyperparameters and per-step optimization, which can lead to overly conservative behavior, infeasibility near tight input limits, and difficulty satisfying long-horizon STL tasks. To address these limitations, we propose a feasibility-aware learning framework that constructs trainable, time-varying High Order Control Barrier Function (HOCBF) constraints and hyperparameters that guarantee satisfaction of a given STL specification. We introduce a unified robustness measure that jointly captures STL satisfaction, constraint feasibility, and control-bound compliance, and propose a neural network architecture to generate control inputs that maximize this robustness. The resulting controller guarantees STL satisfaction with strictly feasible HOCBF constraints and requires no manual tuning. Simulation results demonstrate that the proposed framework maintains high STL robustness under tight input bounds and significantly outperforms fixed-parameter and non-adaptive baselines in complex environments.

\end{abstract}

\section{INTRODUCTION}
Autonomous and robotic systems are often required to meet mission objectives that extend well beyond basic stability or invariance requirements. In practical scenarios such as surveillance, for instance, an unmanned aircraft may need to periodically collect information from a designated area, revisit a charging station for a minimum duration at regular intervals, and consistently steer clear of restricted airspace.
To formally capture such time- and event-dependent requirements, temporal logics—most notably Linear Temporal Logic (LTL) \cite{baier2008principles} and Signal Temporal Logic (STL) \cite{maler2004monitoring}—have become standard specification languages due to their rich expressivity and easily interpretable semantics.

In this work, we address the control of dynamical systems subject to requirements expressed in STL, which provides a rich language for specifying time-dependent requirements over real-valued signals. STL supports both Boolean semantics, which indicate whether a trajectory satisfies a formula, and quantitative (robustness) semantics \cite{donze2010robust}, which assign a real value reflecting the degree of satisfaction or violation. Prior research has shown that enforcing STL constraints can be formulated as an optimization problem in which the robustness metric appears either in the objective or in the constraints. These formulations have been tackled using Mixed Integer Linear Programming (MILP) \cite{raman2014model,sadraddini2015robust} as well as gradient-based optimization methods \cite{pant2017smooth,mehdipour2019arithmetic,gilpin2020smooth}. Although effective, these approaches require considerable computation, limiting their practicality for real-time or online control.

Recent works have also employed Control Barrier Functions (CBFs) to enforce STL specifications on system trajectories \cite{lindemann2018control,garg2019control,xiao2021high2,liu2024auxiliary}. CBFs have been extensively used in the controls community to guarantee safety by ensuring forward invariance of a prescribed safe set \cite{ames2016control,ames2014control}, and they can be combined with Control Lyapunov Functions (CLFs) to additionally encode stability objectives. In both cases, the CBF and CLF appear as inequality constraints on the control input, and the resulting controllers are typically synthesized through quadratic programs (QPs), which are computationally efficient and suitable for real-time implementation. Building on this foundation, the CBF framework has been extended to address higher relative-degree constraints \cite{xiao2021high,nguyen2016exponential}, mixed relative-degree constraints \cite{liu2025auxiliary} and adaptive control \cite{xiao2021adaptive,liu2023auxiliary}. However, since the CBFs or hyperparameters involved in CBFs are typically selected manually, poor selection can lead to excessive conservativeness, and the resulting QPs may become infeasible due to conflicts between the CBF constraints and the input bounds. Although online adjustment of these hyperparameters can improve feasibility, it adds computational overhead and compromises practicality for real‐time applications \cite{liu2025auxiliary,rabiee2025soft,parwana2025rate,breeden2023compositions}.

Learning-based methods can perform most of the computation offline during training, which enables real-time control during execution. In contrast to CBF–QP controllers, many STL-guided learning approaches do not formulate control synthesis as a constrained optimization problem either during training or at run time. For example, model-based Reinforcement Learning (RL) learns a policy by maximizing STL robustness \cite{yaghoubi2019worst,leung2022semi}, and Q-learning has also been applied to STL control synthesis in \cite{aksaray2016q,venkataraman2020tractable}. In addition, the authors of \cite{kalagarla2021model} formulate the STL control synthesis as a constrained Markov Decision Process (cMDP), which yields a computable lower bound on the STL satisfaction probability. More recently, the authors of \cite{meng2025tgpo} propose Temporal Grounded Policy Optimization (TGPO), which learns STL-satisfying policies by decomposing the temporal specification into stage-wise rewards and training a hierarchical RL policy accordingly. Since these policies are trained solely through robustness-based objectives and are directly executed after training, they do not exhibit numerical infeasibility caused by conflicting constraints. Nevertheless, these learning-based approaches cannot guarantee that the resulting policy will satisfy the specification. Violations may occur for two reasons. First, neural network training may converge to a suboptimal local solution under complex dynamics or specifications, leading to undesired behavior. Second, even if a policy satisfies the specification on training data, it may still fail for unseen initial conditions or environments.

To better ensure STL satisfaction, recent works integrate CBFs into learning-based controllers \cite{liu2021recurrent, liu2023safe}. In both these works, however, CBF constraints can only guarantee the ``globally" subformulas of an STL specification, and they are not enforced during policy optimization: \cite{liu2021recurrent} trains an RNN to imitate trajectories that satisfy STL specifications generated offline under CBF constraints, while \cite{liu2023safe} uses CBFs only to guarantee constraint-compliant data collection when learning a probabilistic model. In both cases, the learned policy itself does not encode STL correctness and must be corrected at execution via a runtime CBF–QP filter. This separation creates a mismatch between training and execution, making STL satisfaction dependent on the feasibility of the online QP, which may fail when CBF-based constraints conflict with system dynamics or input bounds. Conversely, directly enforcing CBF constraints during training often makes the optimization infeasible and prematurely halts learning. Thus, a key challenge is to enforce CBF constraints consistently during both training and execution while accounting for conflicts between constraint enforcement and system or input limits to preserve feasibility and maintain STL satisfaction.

In this paper, we guarantee STL satisfaction by enforcing CBF-based constraints during both training and execution. We provide a general, algorithmic procedure to generate these CBFs given an STL formula. We adopt a model-based RL framework in which trajectory rollouts are generated by repeatedly solving a differentiable QP (dQP) whose last layer augments a neural network controller with CBF constraints, incorporating specification correctness directly into policy learning rather than handling it externally. We train three neural networks to generate time-varying hyperparameters in the dQP cost and constraints: (i) InitNet, which initializes hyperparameters based on the initial condition; (ii) RefNet, which updates cost-related hyperparameters along the rollout; and (iii) an extended BarrierNet architecture, inspired by \cite{xiao2021barriernet}, which adapts the HOCBF constraint hyperparameters. These networks are optimized via a unified robustness metric that jointly captures STL satisfaction and QP feasibility, improving performance while accounting for potential infeasibility. After training, the networks provide hyperparameters online for execution, enabling specification-compliant control under new initial conditions. By jointly optimizing specification robustness and feasibility, the resulting controller remains feasible even under tight input bounds or complex STL tasks. Overall, the method avoids manual CBF tuning, reduces over-conservativeness via hyperparameter learning, and enables feasible STL enforcement during both training and execution without requiring pre-generated reference inputs as in \cite{xiao2021barriernet}.

This work substantially extends our prior conference paper \cite{liu2023learning}, which enforced CBF constraints in a dQP using BarrierNet to improve model-based RL for STL satisfaction. The previous formulation optimized only STL robustness and ignored dQP feasibility during training, often causing infeasible QP constraints under tight input limits or complex dynamics, and prematurely terminating learning. It also used CBF hyperparameters that remained constant within each rollout, limiting their ability to adapt over the time horizon and restricting robustness improvement. In this paper, we introduce a unified robustness measure that incorporates both STL satisfaction and QP feasibility, guiding learning toward solutions that remain feasible even in challenging scenarios. In addition, the CBF hyperparameters become time-varying along each rollout, enabling adaptive hyperparameter tuning and further reducing conservativeness while improving robustness. Simulation results further demonstrate that the proposed method achieves higher robustness under strict input limits and challenging STL specifications.

\section{Preliminaries}
\label{sec:prelim}
Consider a nonlinear control-affine system:
\begin{equation}
    \label{eq:system1}
    \mathbf{\dot x} = f(\mathbf x) + g(\mathbf x)\mathbf u,
\end{equation}
where $\mathbf x(t)\in\mathbb R^n$ denotes the state and $\mathbf u(t)\in\mathcal U\subseteq\mathbb R^q$ denotes the control input. The functions $f:\mathbb R^n\rightarrow \mathbb R^n$ and $g:\mathbb R^n\rightarrow \mathbb R^{n\times q}$ are assumed to be locally Lipschitz. The input set $\mathcal U$ is a compact hyper‐rectangle, i.e., $\mathbf u_{min}\leq\mathbf u\leq\mathbf u_{max}$,
where the inequality is interpreted element‐wise. Without loss of generality, we set the initial time to $0$. The initial condition $\mathbf x(0)=\mathbf x_0$ is sampled from a set $\mathcal X_0\subseteq \mathbb R^n$ according to a probability density function $P:\mathcal X_0\rightarrow \mathbb R$. We consider system trajectories over a compact time horizon $[0,T]$. A signal $\mathbf x:[0,T]\rightarrow \mathbb R^n$ is a solution of system \eqref{eq:system1} if it is absolutely continuous and satisfies the dynamics for all $t\in[0,T]$. A partial trajectory over $[0,t]$ is denoted as $\mathbf x_{0:t}$.

We employ a state‐feedback neural network controller,
\begin{equation}
\label{eq:nn-input}
    \mathbf{u}(t) = \pi(\mathbf{x}(t),\bm\theta),
\end{equation}
where $\bm\theta$ represents the network parameters. Memory can be incorporated as
\begin{equation}
\label{eq:nn-input2}
    \mathbf{u}(t)=\pi(\mathbf{x}_{0:t},\boldsymbol{\theta}),
\end{equation}
where the controller depends on state histories and can be implemented using a recurrent architecture such as a Recurrent Neural Network (RNN) \cite{goodfellow2016deep}.
\subsection{Signal Temporal Logic (STL)}
Signal Temporal Logic (STL) \cite{maler2004monitoring} is defined over real-valued signals $\mathbf{x}:\mathbb{R}_{\ge 0}\rightarrow \mathbb{R}^n$, such as trajectories generated by \eqref{eq:system1}. In this work, we focus on a fragment of STL characterized by the following syntax: 
\begin{subequations}
\label{eq:stl}
\begin{align}
\label{eq:stl1}
    \phi &\coloneqq \top\ |\ \mu\ |\ \neg \mu\ |\ \phi_1\land\phi_2,\\
\label{eq:stl2}
    \varphi &\coloneqq F_{[t_a,t_b]}\phi\ |\ G_{[t_a,t_b]}\phi\ |\  \varphi_1\land\varphi_2,
\end{align}
\end{subequations}
where $\phi$ and $\varphi$ denote STL formulae; $\phi_1$ and $\phi_2$ are formulae of class $\phi$, and $\varphi_1$ and $\varphi_2$ are formulae of class $\varphi$; $\top$ denotes the logical \emph{true}; $\mu$ represents an atomic predicate of the form $h(\mathbf{x})\ge 0$, where $h:\mathbb{R}^n\rightarrow\mathbb{R}$; $\neg$ and $\land$ denote Boolean negation and conjunction, respectively; $F$ and $G$ are the temporal operators \emph{eventually} and \emph{always}, respectively; and $[t_a,t_b]$ is a time interval such that $t_a < t_b$. 

We write $(\mathbf x, t)\models \varphi$ to indicate that the signal $\mathbf x$ satisfies the formula $\varphi$ at time $t$. A formal definition of the qualitative semantics of STL is provided in \cite{maler2004monitoring}. Informally, $F_{[t_a,t_b]}\phi$ holds if “$\phi$ becomes \textit{true} at some time in $[t_a,t_b]$”, whereas $G_{[t_a,t_b]}\phi$ holds if “$\phi$ is \textit{true} at all times in $[t_a,t_b]$”. Boolean operators are interpreted in the standard way.
Compared with the full STL syntax in \cite{maler2004monitoring}, the fragment in \eqref{eq:stl} excludes the temporal \textit{until} operator, nested temporal operators such as “\textit{eventually always}”, and Boolean disjunction $\vee $ in order to be able to convert the STL formula into CBFs. Despite this restriction, the fragment remains expressive enough to capture a broad class of practically relevant temporal requirements, including safety and reachability specifications with explicit timing constraints.

STL is also equipped with quantitative semantics, known as \emph{robustness}, which assigns a real value that measures the degree to which a signal satisfies a formula $\varphi$. Several robustness metrics have been proposed in the literature \cite{donze2010robust,mehdipour2019arithmetic,gilpin2020smooth,varnai2020robustness}. In this work, we adopt the smooth robustness introduced in \cite{liu2023robust}, which is differentiable almost everywhere and can be readily embedded in learning-based algorithms. This robustness measure is sound, in the sense that its value is positive if and only if the STL formula is satisfied. We denote the robustness of $\varphi$ at time $t$ with respect to a signal $\mathbf{x}$ by $\rho(\varphi,\mathbf{x},t)$.
We also define the time horizon of an STL formula $\varphi$, denoted by $hrz(\varphi)$, as the earliest future time required to determine its satisfaction and robustness. Throughout this paper, we evaluate trajectories of \eqref{eq:system1} only over the time horizon of the specification, i.e., we set $T = hrz(\varphi)$.

\subsection{Time-Varying High Order Control Barrier Function}
In this subsection, we introduce time-varying High Order Control Barrier Functions (HOCBFs) \cite{xiao2021high2}. We begin by recalling the definition of a class $\kappa$ function:
\begin{definition}[Class $\kappa$ function~\cite{Khalil:1173048}]
 A continuous function $\alpha:[0,a)\to[0,+\infty],a>0$ is called a class $\kappa$ function if it is strictly increasing and $\alpha(0)=0$.
\end{definition}

\begin{definition} The relative degree of a differentiable function $b:\mathbb{R}^{n}\to\mathbb{R}$ is the number of times we need to differentiate it along system dynamics \eqref{eq:system1} until any component of $\mathbf{u}$ explicitly shows in the corresponding derivative. 
\end{definition}

Consider a time-varying constraint $b(\mathbf{x},t)\ge 0$, where $b:\mathbb{R}^n\times[0,T]\rightarrow\mathbb{R}$ is a differentiable function with relative degree $m$. Let $\psi_0(\mathbf{x},t)\coloneqq b(\mathbf{x},t)$. We then construct a sequence of functions $\psi_i:\mathbb{R}^n\times[0,T]\rightarrow\mathbb{R}$, for $i=1,\ldots,m$, defined recursively as follows:
\begin{equation}
\label{eq: psi}
    \psi_i(\mathbf x,t) \coloneqq \dot \psi_{i-1}(\mathbf x,t) + \alpha_i\big(\psi_{i-1}(\mathbf x,t)\big),
\end{equation}
where $\alpha_{i}(\cdot ),\ i\in \{1,...,m\}$ denotes a $(m-i)^{th}$ order differentiable class $\kappa$ function.  A sequence of sets $\mathcal C_{i}(t)$ are then defined based on \eqref{eq: psi} as 
\begin{equation}
\label{eq: set}
    \mathcal C_i(t) = \{\mathbf x\in \mathbb R^n|\psi_i(\mathbf x,t)\geq 0\}.
\end{equation}
In practice, the class $\kappa$ functions $\alpha_i(\cdot)$ are parameterized by a set of scalar hyperparameters $p_i$ that regulate how aggressively the associated functions are enforced, e.g., through a simple scaling of the form $p_i \alpha_i(\cdot)$.
\begin{definition} A set $\mathcal C(t)\subset \mathbb{R}^{n}$ is forward invariant for system \eqref{eq:system1} if its solutions for some $\mathbf{u} \in \mathcal U$ starting from any $\mathbf{x}(0) \in \mathcal C(0)$ satisfy $\mathbf{x}(t) \in \mathcal C(t), \forall t \ge 0.$
\end{definition}

\begin{definition}[Time-varying High Order Control Barrier Function~\cite{xiao2021high}]
\label{def: HOCBF}
Let $\psi_1(\mathbf x,t),\ldots,\psi_m(\mathbf x,t)$ be defined by \eqref{eq: psi} and $\mathcal C_1(t),\ldots,C_m(t)$ be defined by \eqref{eq: set}. A differentiable function $b(\mathbf x,t)$ is a High Order Control Barrier Function (HOCBF) with relative degree $m$ with respect to  \eqref{eq:system1} if there exists differentiable class $\kappa$ functions $\alpha_i$, $i=1,\ldots,m$, such that
\begin{equation}
\label{eq: hocbf}
\begin{aligned}
    \sup_{\mathbf u\in\mathcal U}\big[&L_f^m b(\mathbf x,t) + L_gL_f^{m-1}b(\mathbf x,t)\mathbf u + \frac{\partial^mb(\mathbf x,t)}{\partial t^m}\\
    &+O(b(\mathbf x,t)) + \alpha_m(\psi_{m-1}(\mathbf x,t))\big] \geq 0,
\end{aligned}
\end{equation}
for all $(\mathbf x,t)\in \mathcal C_1(t)\cap\mathcal C_2(t)\cap\ldots\cap\mathcal C_m(t)\times[0,T]$. In \eqref{eq: hocbf}, $L_{f}^{m}$ denote $m$-th order Lie derivatives along $f$; $L_{g}$ is the Lie derivative along $g$. Here $O(\cdot)$ collects the remaining Lie derivatives along $f$ and partial derivatives with respect to $t$ of degree less than $m$, as well as mixed derivatives of total order less than or equal to $m$. Note that the mixed derivative terms vanish when $b(\mathbf x,t)$ is additively separable, i.e., when $\frac{\partial^2b(\mathbf x,t)}{\partial x\partial t}=0$, a condition satisfied by the CBFs considered in this paper. Here, We assume that $L_{g}L_{f}^{m-1}b(\mathbf x,t)\mathbf u\ne0$ on the boundary of set $\mathcal C_{1}(t)\cap...\cap \mathcal C_{m}(t)$. 
\end{definition}
\begin{theorem}[Safety Guarantee~\cite{xiao2021high}]
\label{thm:hocbf}
Given an HOCBF $b(\mathbf x,t)$ with a sequence of sets $\mathcal C_1(t),\ldots,C_m(t)$ as defined in \eqref{eq: set}, if $\mathbf x(0)\in \mathcal C_1(0)\cap\mathcal C_2(0)\cap\ldots\cap\mathcal C_m(0)$, then any Lipschitz continuous controller $\mathbf u(t)$ that satisfies \eqref{eq: hocbf} $\forall t\in[0,T]$ renders $\mathcal C_1(t)\cap\mathcal C_2(t)\cap\ldots\cap\mathcal C_m(t)$ forward invariant for system \eqref{eq:system1}.
\end{theorem}

\section{Problem Formulation and Approach}
\label{sec:prob-form}
Let $J(\mathbf{u})$ denote a cost function defined over control signals 
$\mathbf{u}:[0,T]\rightarrow \mathcal{U}$. We consider the following problem:
\begin{problem}
\label{pb:1}
Given a system with dynamics \eqref{eq:system1}, an STL specification $\varphi$ in \eqref{eq:stl}, and an initial state $\mathbf{x}_0$ sampled from the distribution $P:\mathcal X_0\rightarrow\mathbb{R}$, the goal is to find an optimal control input $\mathbf{u}^*(t)$ that maximizes STL robustness $\rho(\varphi,\mathbf x, 0)$ while minimizing the cost $J(\mathbf{u})$, subject to satisfying $\varphi$:
\begin{subequations}
    \label{eq:goal}
    \begin{align}
       \mathbf u^*(t) = &\arg\max_{\mathbf u(t)} \rho(\varphi,\mathbf x, 0) - J(\mathbf u) \\
       \text{s.t.}\quad & \dot{\mathbf x} = f(\mathbf x) + g(\mathbf x)\mathbf u(t), \\
       & \mathbf u_{min} \leq \mathbf u(t) \leq \mathbf u_{max},\\
       & (\mathbf x,0)\models \varphi\label{subeq:stl fml}.
    \end{align}
\end{subequations}
\end{problem}

STL formulas provide a unified and expressive framework
to specify both safety requirements (e.g., obstacle avoidance) and goal-reaching
tasks (e.g., reach within a specified time interval)
over system trajectories. Existing works~\cite{nguyen2016exponential, xiao2021high} employ HOCBFs~\eqref{eq: hocbf} to address safety specifications in \eqref{subeq:stl fml} with high relative degree by embedding them as constraints in quadratic optimization-based controllers. In these approaches, HOCBF constraints are enforced as strict hard constraints to guarantee safety, which implicitly requires the QP to remain feasible at every sampling instant. However, feasibility is not guaranteed and is often violated when the control bounds $\mathcal U$ are tight or when the HOCBF hyperparameters (e.g., the coefficients in the class $\kappa$ functions) are manually designed and fixed. In particular, based on \eqref{eq: psi}, fixed hyperparameters force the constraint $\psi_{i}(\mathbf x,t)\ge 0$ to impose a constant decay rate on function $\psi_{i-1}(\mathbf x,t)$, regardless of how close the system actually is to the boundary. This leads to overly aggressive constraint satisfaction in non-critical situations, resulting in unnecessarily conservative control actions and reduced performance. In \cite{xiao2021barriernet}, BarrierNet was introduced as a differentiable QP layer with HOCBF constraints appended to the output of a neural network controller to guarantee safety. In this framework, class $\kappa$ hyperparameters are learned during training, which substantially reduces conservativeness compared with fixed designs. However, the approach relies on supervised learning, where a reference control input that already satisfies safety and input bounds must be provided during training. Such a reference input is often difficult to obtain, and becomes especially challenging when the task is governed by complex STL specifications. Therefore, there is a strong need for methods that integrate BarrierNet with STL satisfaction while automatically generating reference inputs, rather than requiring them to be specified a priori.

\textbf{Approach:} In this paper, we train a neural network controller, with or without memory as in \eqref{eq:nn-input2} and \eqref{eq:nn-input}, to solve the optimization Problem~\ref{pb:1}. We address STL satisfaction by enforcing trainable HOCBF‐based constraints during both training and execution. From the STL formula $\varphi$, we construct time‐varying HOCBFs and embed them into a neural network controller via a BarrierNet layer \cite{xiao2021barriernet}, enabling specification correctness and robustness to be incorporated directly into policy learning. The HOCBF hyperparameters are generated by three networks: InitNet, which adapts them to different initial conditions; RefNet, which automatically produces reference controls; and extended BarrierNet, inspired by \cite{xiao2021barriernet}, which adjusts CBF hyperparameters to enforce $\varphi$. These components are jointly optimized using a unified robustness metric that accounts for both STL satisfaction and QP feasibility, allowing the resulting controller to remain feasible even under tight input bounds or complex specifications. After training, the networks provide hyperparameters online, ensuring feasible STL enforcement without manual tuning or pre‐generated reference inputs, thereby reducing conservativeness and improving robustness in both training and execution. An overview of the proposed approach is illustrated in Fig.~\ref{fig:framework}, while additional
implementation details are provided in Sec.~\ref{sec:FeasiBarrierNet}.

\section{Learning Controllers from Feasibility-Aware STL via BarrierNet} 
\label{sec:FeasiBarrierNet}

In this section, we present our solution to Problem 1. Section \ref{subsec:barriernet} introduces the trainable HOCBFs and the state-dependent BarrierNet architecture extended from [19], along with a sufficient condition for guaranteeing QP feasibility. Section \ref{subsec:hocbf-stl} provides a general procedure for constructing a set of time-varying HOCBFs that guarantee the satisfaction of a given STL specification. Section \ref{subsec:Unified Rob} then defines a unified robustness metric that captures both STL satisfaction and QP feasibility. Finally, Section \ref{subsec:learning} describes how these time-varying HOCBFs are jointly trained with the neural network controller through BarrierNet to further enhance robustness.
\subsection{Extended BarrierNet with Trainable HOCBFs and Feasibility Guarantees}
\label{subsec:barriernet}
Suppose we have a set of time-varying HOCBFs 
$b_j(\mathbf{x},t,\boldsymbol{\theta}_b,\mathbf{x}_0)$, $j=1,\ldots,M$, 
each depending on the initial condition $\mathbf{x}_0$ and containing 
trainable parameters $\boldsymbol{\theta}_b$ 
(the role of $\mathbf{x}_0$ will be clarified in Section~\ref{subsec:learning}).
To reduce conservativeness, the associated class $\kappa$ functions 
are also made trainable. Rewriting~\eqref{eq: psi} for a HOCBF $b_j$ yields:
\begin{equation}
\begin{aligned}
    \label{eq: psi2}
    \psi_{i,j}(\mathbf x,t,&\bm\theta_b,\bm\theta_p,\mathbf{x}_0,\mathbf{P}_{\text{inip}}) \coloneqq \dot \psi_{i-1,j}(\mathbf x,t,\bm\theta_b,\bm\theta_p,\mathbf{x}_0,\mathbf{P}_{\text{inip}}) + \\&p_{i,j}(\mathbf x,\bm\theta_p,\mathbf{P}_{\text{inip}})\alpha_{i,j}\big(\psi_{i-1,j}(\mathbf x,t,\bm\theta_b,\bm\theta_p,\mathbf{x}_0,\mathbf{P}_{\text{inip}})\big),
\end{aligned}
\end{equation}
where $\alpha_{i,j}$ are class $\kappa$ functions, 
$+\infty >p_{i,j}(\mathbf{x},\boldsymbol{\theta}_p,\mathbf{P}_{\text{inip}})>0$ for 
$i=1,\ldots,m$ and $j=1,\ldots,M$, and $m$ denotes the 
relative degree of the HOCBF $b_j$. For notational simplicity, we write 
$\psi_{\cdot,j}(\mathbf{x},t,\boldsymbol{\theta}_b,\boldsymbol{\theta}_p)$ 
to denote 
$\psi_{\cdot,j}(\mathbf{x},t,\boldsymbol{\theta}_b,\boldsymbol{\theta}_p,\mathbf{x}_0,\mathbf{P}_{\text{inip}})$ 
throughout the remainder of this paper.
Each $p_{i,j}(\mathbf{x},\boldsymbol{\theta}_p,\mathbf{P}_{\mathrm{inip}})$ is time-varying, as it depends on the state $\mathbf{x}(t)$ and the trainable parameter $\boldsymbol{\theta}_p$. Its initial value $\mathbf{P}_{\mathrm{inip}}(\mathbf{x}_0,\boldsymbol{\theta}_{\mathrm{inip}})$ depends on the initial condition and the trainable parameter $\boldsymbol{\theta}_{\mathrm{inip}}$
(This will also be clarified in 
Section~\ref{subsec:learning}).

BarrierNet~\cite{xiao2021barriernet} is a neural network layer implemented 
as a differentiable Quadratic Program (dQP) with HOCBF constraints. 
We append an extended BarrierNet layer, extending the formulation in 
\cite{xiao2021barriernet}, as the terminal layer of the neural network 
controller (with or without memory), as shown in~\eqref{eq:nn-input2} and 
\eqref{eq:nn-input}, where $\mathbf{u}^*(t)$ is obtained from:
\begin{subequations}
    \label{eq:barriernet}
    \begin{align}
    &\mathbf u^*(t) \hspace{-2pt} = & & \hspace{-10pt} \arg\min_{\mathbf u(t)}\frac{1}{2}\mathbf  u(t)^\top \mathbf Q(\cdot,\bm\theta_q) \mathbf u(t) + \mathbf F^\top(\cdot,\bm\theta_f) \mathbf u(t)    \\
    & \text{s.t.}& & \hspace{-10pt} L_f^m b_j(\mathbf x,t,\bm\theta_b,\mathbf x_0) + L_gL_f^{m-1}b_j(\mathbf x,t,\bm\theta_b,\mathbf x_0)\mathbf u(t) \nonumber \label{subeq:dqp-hocbf}\\ 
    & & & \hspace{-10pt} + \frac{\partial^mb_j(\mathbf x,t,\bm\theta_b,\mathbf x_0)}{\partial t^m} +O(b_j(\mathbf x,t,\bm\theta_b,\mathbf x_0),\bm\theta_p,\mathbf{P}_{\text{inip}}) \nonumber \\ 
    & & & \hspace{-10pt} + p_{m,j}(\mathbf x,\bm\theta_p,\mathbf{P}_{\text{inip}})\alpha_m(\psi_{m-1,j}(\mathbf x,t,\bm\theta_b,\bm\theta_p))\geq 0,\\
    & & &\hspace{-10pt}  \mathbf u_{min} \leq \mathbf u(t) \leq \mathbf u_{max},\label{subeq:input-bound}\\
    & & &\hspace{-10pt} t=k\Delta t,\ k=0,1,2,\ldots,\ j=1,\ldots,M,\label{subeq:discrete-form}
    \end{align}
\end{subequations}
where $\mathbf{Q}(\cdot,\boldsymbol{\theta}_q)\in\mathbb{R}^{q\times q}$, 
$\mathbf{F}(\cdot,\boldsymbol{\theta}_f)\in\mathbb{R}^{q}$, 
$+\infty >p_{i,j}(\mathbf{x},\boldsymbol{\theta}_p,\mathbf{P}_{\text{inip}})>0$ for $i=1,\ldots,m$, 
and $b_j(\mathbf{x},t,\boldsymbol{\theta}_b,\mathbf{x}_0)$ are all generated by 
preceding neural network layers with trainable parameters 
$(\boldsymbol{\theta}_q,\boldsymbol{\theta}_f,\boldsymbol{\theta}_b,\boldsymbol{\theta}_{\text{inip}},\boldsymbol{\theta}_p)\coloneqq\boldsymbol{\theta}$. Here, the symbol ``$\cdot$'' indicates the input to the network:
$\mathbf{x}(t)$ in the memoryless case, and the state history
$\mathbf{x}_{0:t}$ when memory is used.
The matrix $\mathbf{Q}$ is positive definite, and 
$\mathbf{Q}^{-1}\mathbf{F}$ can be interpreted as a reference control. 
Although $\mathbf{Q}$, $\mathbf{F}$, and $p_{i,j}$ are generated by previous layers 
in~\eqref{eq:barriernet}, they may also be treated as directly trainable hyperparameters. The dQP \eqref{eq:barriernet} is solved at each time point $k\Delta t$, $k=0,1,\ldots$ until reaching the time horizon $T$, and the solution $\mathbf u^*(t)$ is applied to the system as a constant for the time period $[k\Delta t, k\Delta t + \Delta t)$. 
\begin{theorem}
\label{thm:dqp-hocbf}
Given a time-varying HOCBF $b(\mathbf{x}, t)$ with a sequence of sets 
$\mathcal{C}_1(t), \ldots, \mathcal{C}_m(t)$ defined in \eqref{eq: set}, 
assume that the initial condition satisfies $
\mathbf{x}(0) \in \mathcal{C}_1(0) \cap \mathcal{C}_2(0) \cap \cdots \cap \mathcal{C}_m(0).$
Let $p_{i,j}(\mathbf{x}, \boldsymbol{\theta}_p, \mathbf{P}_{\text{inip}})$,
$i \in \{1, \ldots, m-2\}$ be differentiable functions chosen such that
all time derivatives appearing in the recursion up to $\psi_{m-1,j}$
are independent of the control input $\mathbf{u}$. Specifically, we require
\begin{equation} \label{eq:order-rule} \begin{aligned} L_g L_f^k p_{i,j}(\mathbf{x},\boldsymbol{\theta}_p,\mathbf{P}_{\text{inip}}) = 0,~ k = m-2-i, \\m\ge 3,~ \forall i=1,\ldots,m-2. \\ \end{aligned} \end{equation}
Under this condition, the Lipschitz continuous control input $\mathbf{u}$
enters the HOCBF constraint only through the highest-order term
$\psi_{m,j}(\mathbf{x}, t, \boldsymbol{\theta}_b, \boldsymbol{\theta}_p)$
in \eqref{subeq:dqp-hocbf}, subject to the input bound \eqref{subeq:input-bound}.
As a result, forward invariance of the set intersection defined in \eqref{eq: set}
is guaranteed regardless of whether the functions
$p_{i,j}$
depend on the state.
In particular, when $p_{i,j}$ are treated as trainable hyperparameters that do not
depend on the state, the resulting BarrierNet composed of neurons defined in
\eqref{eq:barriernet} also guarantees forward invariance.
\end{theorem}
\begin{proof}
We show that each constraint \eqref{subeq:dqp-hocbf} has the HOCBF form, thus forward invariance follows from Theorem \ref{thm:dqp-hocbf}.
\begin{figure}
    \centering
\includegraphics[width=8cm]{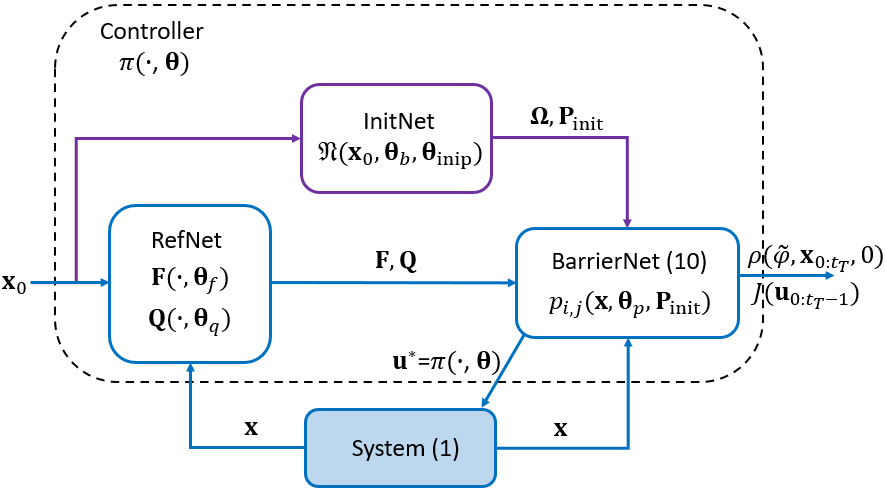}
    \caption{Overall structure of the controller. The purple module runs only at $t=0$, while
the blue modules operate at every later time step. The dashed box outlines the
controller $\pi(\cdot,\boldsymbol{\theta})$ (\eqref{eq:nn-input} or \eqref{eq:nn-input2}) over one rollout (epoch).}
    \label{fig:framework}
\end{figure}
Assume that $p_{i,j}(\mathbf{x},\boldsymbol{\theta}_p,\mathbf{P}_{\text{inip}})$, $i\in\{1,\dots,m-2\}$ are differentiable and satisfy
\eqref{eq:order-rule} (the multiplier $p_{m-1,j}(\mathbf{x},\boldsymbol{\theta}_p,\mathbf{P}_{\text{inip}})$ can introduce a dependence on $\mathbf{u}$ only in $\psi_{m,j}(\mathbf{x},t,\boldsymbol{\theta}_b,\boldsymbol{\theta}_p)$), or are trainable hyperparameters. By repeatedly applying the modified
recursion in \eqref{eq: psi2}, the highest--order constraint $\psi_{m,j}(\mathbf{x},t,\boldsymbol{\theta}_b,\boldsymbol{\theta}_p)$ associated with
$b_j(\mathbf x,t,\bm\theta_b,\mathbf x_0)$ takes the form
\begin{equation}
\label{eq:psi_m_expanded-g}
\begin{aligned}
&\psi_{m,j}
= L_f^{m} b_j
+ \big[L_g L_f^{m-1} b_j\big]\mathbf u+ \frac{\partial^mb_j}{\partial t^m} +\\
& \Phi_{j}\big(b_j, \dot b_j, \dots, b_j^{(m-1)},
               p_{1,j}, \dot p_{1,j},\ddot{p}_{1,j},\dots, \dot p_{m-1,j}, p_{m,j}\big),
\end{aligned}
\end{equation}
where $\dot b_j$ and $b_j^{(k)}$ denote the total time derivatives
$\frac{d^k}{dt^k} b_j(\mathbf{x},t)$, which include both
$\frac{\partial b_j}{\partial \mathbf{x}}\dot{\mathbf{x}}$ and
$\frac{\partial b_j}{\partial t}$.
The term $\frac{\partial^m b_j}{\partial t^m}$ is written explicitly,
while all remaining mixed terms, including the time derivatives of the
multipliers $p_{i,j}$, are collected in $\Phi_{j}$. Here, for notational simplicity, the arguments $(\mathbf{x},t,\boldsymbol{\theta}_b,\boldsymbol{\theta}_p,\mathbf{x}_0,\mathbf{P}_{\text{inip}})$ of 
$p_{i,j}$, $b_j$, $\Phi_{j}$, and $\psi_{m,j}$ have been omitted. Specifically, if  the class $\kappa$ functions $\alpha_{i,j}(\cdot)$ are linear, we obtain
\begin{equation}
\label{eq:psi_m_expanded}
\begin{aligned}
&\psi_{m,j}
= L_f^{m} b_j
+ \big[L_g L_f^{m-1} b_j\big]\mathbf u+ \frac{\partial^mb_j}{\partial t^m} +\\
& \sum_{r=1}^{m}
\Bigg[
\sum_{1\le k_1 < \cdots < k_r \le m}
\Big( \prod_{\ell=1}^{r} p_{k_\ell,j} \Big)
+ P_{r,j}(p_{i,j})
\Bigg] b_j^{(m-r)},
\end{aligned}
\end{equation}
where $P_{r,j}(p_{i,j})$ collects polynomial expressions in the time derivatives of
the multipliers $p_{i,j}$. 

Because the system dynamics are control--affine and the functions $p_{i,j}$
are differentiable, every time derivative of $p_{i,j}$ appearing in $\Phi_{j}$ (and $P_{r,j}(p_{i,j})$)
is at most affine in $\mathbf{u}$. Moreover, condition \eqref{eq:order-rule} ensures that at most one
derivative of $p_{i,j}$ yields a nonzero input gain inside $\Phi_{j}$.
Therefore, each $\Phi_{j}$ can be decomposed as
\begin{equation}
\label{eq:p_expanded}
\Phi_{j} = \Phi_{j}^{0}(\mathbf{x},t,\boldsymbol{\theta}_b,\boldsymbol{\theta}_p,\mathbf{x}_0,\mathbf{P}_{\text{inip}}) + \Phi_{j}^{1}(\mathbf{x},t,\boldsymbol{\theta}_b,\boldsymbol{\theta}_p,\mathbf{x}_0,\mathbf{P}_{\text{inip}})\mathbf{u},
\end{equation}
for smooth functions $\Phi_{j}^{0}$ and $\Phi_{j}^{1}$ independent of $\mathbf{u}$.
Substituting this into~\eqref{eq:psi_m_expanded} and regrouping terms gives
\begin{equation}
\label{eq:psi_m_expanded2}
\psi_{m,j}
= a_j\mathbf{u} + c_j,
\end{equation}
where
\begin{equation}
\label{eq:a_expanded}
a_j
= L_g L_f^{m-1} b_j
+\Phi_{j}^{1}(\mathbf{x},t,\boldsymbol{\theta}_b,\boldsymbol{\theta}_p,\mathbf{x}_0,\mathbf{P}_{\text{inip}}),
\end{equation}
and
\begin{equation}
\label{eq:c_expanded}
\begin{aligned}
c_j
= L_f^{m} b_j
+ \frac{\partial^mb_j}{\partial t^m} +\Phi_{j}^{0}(\mathbf{x},t,\boldsymbol{\theta}_b,\boldsymbol{\theta}_p,\mathbf{x}_0,\mathbf{P}_{\text{inip}}).
\end{aligned}
\end{equation}
If $p_{i,j}$ are trainable constants, then $\Phi_{j}^{1}(\mathbf{x},t,\boldsymbol{\theta}_b,\boldsymbol{\theta}_p,\mathbf{x}_0,\mathbf{P}_{\text{inip}})=\mathbf{0}$ (and $P_{r,j}(p_{i,j})=\mathbf{0}$). Thus, in both cases, $\psi_{m,j}$ is affine in $\mathbf{u}$, matching the standard HOCBF
form. Hence each inequality in \eqref{subeq:dqp-hocbf} defines a valid HOCBF constraint.
Since the dQP \eqref{eq:barriernet} yields a Lipschitz continuous control law, Theorem~\ref{thm:hocbf}
implies the forward invariance of all encoded sets \eqref{eq: set}. Stacking these
constraints in the BarrierNet layer guarantees forward invariance of
their intersection. 
\end{proof}
Since \eqref{eq:barriernet} is differentiable, the gradients of the optimal control $\mathbf{u}^*(t)$ and the HOCBF hyperparameters $p_{i,j}$ with respect to $\boldsymbol{\theta}$ can be obtained via the dQP framework of \cite{amos2017optnet} (the gradients of $p_{i,j}$ with respect to $\mathbf{x}$ can also be computed by PyTorch). This enables $\boldsymbol{\theta}$ to be trained with any common neural–network optimizer (e.g., those used for LSTMs or CNNs \cite{goodfellow2016deep}).
As a result, BarrierNet not only guarantees satisfaction of all HOCBF constraints encoding the STL requirements, but also enables the controller to optimize a given objective through training.

In \eqref{eq: psi2}, when the class $\kappa$ function $\alpha_{i,j}(\cdot)$ degenerates to 
zero, the multiplier $p_{i,j}$ loses its ability to regulate the constraint. In this 
case, the dQP may also become infeasible, since each $\psi_{m,j}$ constraint 
can directly conflict with each other and the control bounds, preventing the existence of a 
control input $\mathbf{u}$ that satisfies all constraints simultaneously. 
The following theorem provides a sufficient condition to ensure the feasibility 
of the dQP.

\begin{theorem}[Sufficient Condition for Feasibility]
\label{thm:hierarchical-feasibility}
Consider a time-varying HOCBF $b(\mathbf{x}, t)$ and a sequence of sets
$\mathcal{C}_1(t), \ldots, \mathcal{C}_m(t)$ defined in \eqref{eq: set}.
Assume that the initial condition satisfies
$\mathbf{x}(0) \in \mathcal{C}_1(0) \cap \mathcal{C}_2(0) \cap \cdots \cap \mathcal{C}_m(0).$ For $i \in \{1, \ldots, m-2\}$, let
$p_{i,j}(\mathbf{x}, \boldsymbol{\theta}_p, \mathbf{P}_{\text{inip}})$
be differentiable functions that satisfy \eqref{eq:order-rule}.
Alternatively, these functions may be treated as trainable hyperparameters
that do not depend on the state. For each $j \in \{1, \ldots, M\}$, define the highest-order HOCBF constraint
by \eqref{subeq:dqp-hocbf}, and assume that the control input $\mathbf{u}$
satisfies the bound \eqref{subeq:input-bound}.
Suppose that the condition $\psi_{m-1,j} > 0$ is enforced for all $t \ge 0$. Then, for any $(\mathbf{x}, t)$ such that
$\psi_{m-1,j}(\mathbf{x}, t, \boldsymbol{\theta}_b, \boldsymbol{\theta}_p) > 0$,
the optimization problem in \eqref{eq:barriernet} is pointwise feasible.
As a result, the corresponding dQP admits a feasible control input
$\mathbf{u}^*(t)$.
\end{theorem}
\begin{proof}
The highest-order HOCBF constraint \eqref{subeq:dqp-hocbf} can be rewrite as 
\begin{equation}
\label{eq:high-order-feasibility}
\begin{aligned}
 & L_f^{m} b_j+
\big[L_g L_f^{m-1} b_j\big]\mathbf u+ \frac{\partial^mb_j}{\partial t^m} +\\
&\sum_{i=1}^{m-1}
\alpha_{i,j}(\psi_{i-1,j}) [L_{g} L_f^{m-i-1} p_{i,j}] \mathbf{u}
    + R(b_j,p_{i,j})
 + \\
&\sum_{i=1}^{m-1}
\Big[
    L_f^{m-i} p_{i,j}\alpha_{i,j}(\psi_{i-1,j})
    + p_{i,j}\alpha_{m,j}(\psi_{m-1,j})
\Big]
\;\ge 0.
\end{aligned}
\end{equation}
Here, $R(b_j,p_{i,j})$ collects all remaining terms arising from the recursive
application of the class $\kappa$ functions, the drift dynamics $f(\mathbf{x})$,
and the chain rule. Consequently, $R(b_j,p_{i,j})$ depends on Lie derivatives of
$b_j$ and $p_{i,j}$ along $f$ (up to order $m-1$), as well as additional terms
resulting from the explicit time dependence of $b_j(\mathbf{x},t)$, and does not
explicitly depend on the control input $\mathbf{u}$.
The term $L_{g} L_f^{m-i-1} p_{i,j}\mathbf{u}$ is affine in $\mathbf{u}$
whenever $L_{g} L_f^{m-i-1} p_{i,j}\ne \mathbf{0}$; if
$L_{g} L_f^{m-i-1} p_{i,j}= \mathbf{0}$ (e.g., when $p_{i,j}$ is constant in each $\Delta t$ or does not depend on any state component
directly affected by the control input), this contribution vanishes and
\eqref{eq:high-order-feasibility} still remains affine in $\mathbf{u}$. Moreover, if $\psi_{m-1,j}>0$ for all $t\ge0$, then $\alpha_{m,j}(\psi_{m-1,j})>0$, and the multiplier $p_{m,j}$ can be chosen
so that the term $p_{m,j}\alpha_{m,j}(\psi_{m-1,j})$ dominates the sum of all
other terms in \eqref{eq:high-order-feasibility}, which are jointly bounded
even when $\mathbf{u}$ is restricted by \eqref{subeq:input-bound}. In particular,
even if the total contribution of all terms except $p_{m,j}\alpha_{m,j}(\psi_{m-1,j})$
is negative, one can always increase $p_{m,j}$ to render the left-hand side
of \eqref{eq:high-order-feasibility} nonnegative. Hence, the constraint
\eqref{eq:high-order-feasibility} is pointwise feasible for all
$(\mathbf{x},t)$ with $\psi_{m-1,j}(\mathbf{x},t)>0$,
so the dQP admits a feasible solution $\mathbf{u}^*(t)$ at each time step.
\end{proof}
The condition $\psi_{m-1,j}>0$ in Theorem~\ref{thm:hierarchical-feasibility} is a
sufficient condition for the feasibility of problem~\eqref{eq:barriernet}, and,
together with Theorem~\ref{thm:dqp-hocbf}, guarantees the forward invariance of the
HOCBF-defined set and feasibility of the QP at every time step. From \eqref{eq:high-order-feasibility}, we observe that if
$\psi_{m-1,j} > 0$ cannot be maintained, one can instead impose
$\psi_{m-2,j} > 0$ as the feasibility condition.
In this case, by increasing $L_f p_{m-1,j}$, the left-hand side of
\eqref{eq:high-order-feasibility} can still be rendered positive, ensuring the
feasibility of the constraint.
If $\psi_{m-2,j} > 0$ also cannot be ensured, we may proceed hierarchically and
use $\psi_{m-3,j} > 0$ as the feasibility condition, and so on, down to
$\psi_{0,j} > 0$, i.e., $b_j>0$.

However, the constraint \eqref{subeq:dqp-hocbf} is specifically designed to
guarantee $\psi_{m-1,j} \ge 0$, which is the requirement needed for forward
invariance. Once $\psi_{m-1,j} \ge 0$ cannot be enforced, the constraint
\eqref{subeq:dqp-hocbf} no longer serves its intended purpose.
For this reason, Theorem~\ref{thm:hierarchical-feasibility} adopts only
$\psi_{m-1,j} > 0$ as a sufficient condition for feasibility. From \eqref{eq:order-rule}, we note that the dependence of the
multipliers $p_{i,j}$ on the state $\mathbf{x}$ can be designed so that the
control input appears in lower-order derivatives of $b_j$, effectively reducing
the relative degree. This mechanism complements the hierarchical
feasibility conditions discussed above and will be explored in future work.
\begin{remark}
\label{rem:rem1}
Different from the original BarrierNet \cite{xiao2021barriernet}, which only
considers time-invariant HOCBFs $b_{j}(\mathbf{x})$, the formulation in \eqref{eq:barriernet}
incorporates time-varying HOCBFs $b_{j}(\mathbf{x},t)$. Moreover, we make the HOCBF $b_{j}$
itself trainable, in addition to the hyperparameters $\mathbf{Q}$, $\mathbf{F}$,
and $p_{i,j}$. In Theorem~2, we explicitly allow the time derivatives of the
multipliers $p_{i,j}$ to introduce dependence on the input $\mathbf{u}$ in the
highest--order constraint $\psi_{m,j}$, a feature that was not considered in
\cite{xiao2021barriernet}. Therefore, the BarrierNet structure in (10),
together with Theorem~2, extends the original BarrierNet framework proposed
in \cite{xiao2021barriernet}. In \cite{liu2023learning}, the trainable multiplier 
$p_{i,j}(\mathbf{x}_0,\boldsymbol{\theta}_p)$ remains constant over time 
because it depends only on the initial state, although it varies across 
rollouts. This limits its ability to adapt along the time horizon and 
restricts potential robustness improvement. In contrast, in this paper 
$p_{i,j}(\mathbf{x},\boldsymbol{\theta}_p,\mathbf{P}_{\text{init}})$ is 
time-varying, allowing it to continuously adjust to the evolving system state and thereby better satisfy the STL specification.
\end{remark}
\subsection{HOCBFs for STL specifications}
\label{subsec:hocbf-stl}
The authors of \cite{lindemann2018control} introduced the use of time-varying CBFs to enforce the satisfaction of STL specifications. However, \cite{lindemann2018control} considers only CBFs of relative degree one, and the construction of these CBFs is illustrated through examples rather than provided in a general or systematic form.
In this work, we extend their approach to high-order CBFs and develop a general, algorithmic procedure for constructing such HOCBFs. Moreover, we make these HOCBFs trainable, eliminating the need for manual design and enabling the resulting controller to further improve its performance through training.

Consider an STL formula $\varphi$ of the form \eqref{eq:stl}.
Since any predicate of the form $\neg\mu$ can be rewritten by replacing the
predicate function with $-h(\mathbf{x})$ and eliminating the negation, we
assume, without loss of generality, that $\varphi$ is negation-free.
We impose the following assumption on the STL formula and the system:
\begin{assumption}
\label{as:feasible}
$\forall \mathbf x(0)\in\mathcal X_0$, $\exists \mathbf u(t)\in\mathcal U$ such that $(\mathbf x, 0)\models\varphi$ where $\mathbf x$ is the solution of system \eqref{eq:system1}.
\end{assumption}
Assumption~\ref{as:feasible} is not restrictive in practice, because if it does
not hold, then for certain initial conditions $\mathbf{x}_0$ the specification
$\varphi$ is inherently infeasible and Problem~\ref{pb:1} admits no solution.

\noindent\textbf{Categories of Predicates.} Suppose that there are $M$ predicates in $\varphi$ and they are given by $\mu_j:\ h_j(\mathbf x)\geq0$, $j=1,\ldots,M$. Now we divide all predicates into three categories:
\begin{itemize}
    \item Category I: predicates that are already satisfied at $t=0$ and whose
    enclosing temporal operator also begins at $t=0$.  
    For example, $\mu_1$ in $G_{[0,5]}\mu_1$ and in $G_{[0,5]}(\mu_1 \land \mu_2)$,
    where $h_1(\mathbf{x}_0)\ge 0$.  
    Such predicates typically encode safety requirements, such as obstacle
    avoidance in robotic applications.
    \item Category II: predicates enclosed by $F_{[t_a,t_b]}$ that do not belong
    to Category~I, e.g., $\mu_1$ in $F_{[2,5]}\mu_1$ and $\mu_2$ in
    $F_{[0,5]}(\mu_2 \land \mu_3)$, where $h_2(\mathbf{x}_0)<0$.
    \item Category III: predicates enclosed by $G_{[t_a,t_b]}$ that do not belong
    to Category~I.  
    For example, $\mu_1$ in $G_{[2,5]}(\mu_1 \land \mu_2)$.  
    Note that Assumption~\ref{as:feasible} rules out formulae such as
    $G_{[0,5]}\mu_1$ with $h_1(\mathbf{x}_0)<0$, which are infeasible from the
    initial condition.
\end{itemize}
\noindent\textbf{STL Guarantees.}  
To each predicate $\mu_j$, we associate a time-varying HOCBF $b_j$.  
For predicates in Category~I, which are already satisfied at $t=0$, we assign
a fixed, time-invariant HOCBF to preserve their satisfaction over the required
time interval:
\begin{equation}
\label{eq:cat1}
b_j(\mathbf x) = h_j(\mathbf x).
\end{equation}
For predicates $\mu_j$ in Category II and III, we assign a trainable time-varying HOCBF:
\begin{equation}
    \label{eq:cat23}
    b_j(\mathbf x,t,\bm\theta_b,\mathbf x_0) = h_j(\mathbf x) + \gamma_j (t,\bm\omega_j(\bm\theta_b,\mathbf x_0)),
\end{equation}
where $\gamma_j(\cdot,\boldsymbol{\omega}_j):[0,T]\rightarrow\mathbb{R}$ is a
function parameterized by $\boldsymbol{\omega}_j$, which is generated by a
neural network taking $\mathbf{x}_0$ as input and using parameters
$\boldsymbol{\theta}_b$. Details of this network are provided in
Section~\ref{subsec:learning}.  
For notational simplicity, in the remainder of this subsection we omit
$\boldsymbol{\theta}_b$ and $\mathbf{x}_0$ and treat $\boldsymbol{\omega}_j$ as
a standalone vector.  
With an appropriate choice of $\gamma_j(t,\boldsymbol{\omega}_j)$, enforcing
$b_j(\mathbf{x},t)\ge 0$ for all $t\in[0,T]$ ensures that the predicate
$\mu_j$ is satisfied over its required time interval.  
We now discuss how to select $\gamma_j(t,\boldsymbol{\omega}_j)$.

For notational simplicity, we omit the subscript $j$ when the meaning is clear
from context.  
For a Category~II predicate $\mu$ wrapped by $F_{[t_a,t_b]}$, we choose
$\gamma$ to be a linear function:
\begin{equation}
    \label{eq:gamma_f}
    \gamma(t,\bm\omega) = \omega_1 + \omega_2t,
\end{equation}
where $\boldsymbol{\omega} = [\omega_1,\omega_2]^{\top}$ with $\omega_1 > 0$ and
$\omega_2 < 0$.  Other function forms are also possible.  
To ensure that the HOCBF $b(\mathbf{x},t) = h(\mathbf{x}) + \gamma(t,\boldsymbol{\omega})$
guarantees the satisfaction of $\mu$, we impose the following three
constraints on $\gamma$:
\begin{subequations}
\label{eq:f}
    \begin{align}
    \label{eq:f0}
        &\gamma(0,\bm\omega)\ > -h(\mathbf x_0),\\
    \label{eq:fb}
        &\gamma(t_b, \bm\omega)  < 0,\\
    \label{eq:fa}
        &\gamma(t_a, \bm\omega) > -\sup_{\mathbf x\in\mathbb R^n} h(\mathbf x).
    \end{align}
\end{subequations}
Constraint~\eqref{eq:f0} ensures that the HOCBF is positive at the initial
time, i.e., $b(\mathbf{x}_0,0)>0$.  
Constraint~\eqref{eq:fb} ensures that $h(\mathbf{x}) \ge b(\mathbf{x},t)$ before $t_b$. Under \eqref{eq:f0} and \eqref{eq:fb}, the forward invariance
of the superlevel set of $b(\mathbf{x},t)$ enforces the satisfaction of
$F_{[t_a,t_b]}\mu$.  
The third constraint~\eqref{eq:fa} ensures that the superlevel set of
$b(\mathbf{x},t)$ is nonempty for $t<t_a$.  
As it will be discussed later, we delete this HOCBF once $h(\mathbf{x})\ge0$ for some $t\ge t_a$, so we do not require the superlevel set of $b(\mathbf{x},t)$ to
remain nonempty thereafter.

For a Category~III predicate $\mu$ wrapped by $G_{[t_a,t_b]}$, we define
$\gamma$ as follows:
\begin{equation}
    \label{eq:gamma_g}
    \gamma(t,\bm\omega) = \omega_1 e^{-\omega_2t}-c,
\end{equation}
where $\boldsymbol{\omega} = (\omega_1,\omega_2)$ with $\omega_1>0$ and
$\omega_2>0$, and $c>0$ is a small constant. Other functional forms are also
possible.  
Analogous to \eqref{eq:f}, we impose the following two constraints on
$\gamma$:
\begin{subequations}
\label{eq:g}
    \begin{align}
    \label{eq:g0}
        &\gamma(0,\bm\omega)\ > -h(\mathbf x_0),\\
    \label{eq:ga}
        &\gamma(t_a,\bm\omega)  < 0.
    \end{align}
\end{subequations}
The difference is that \eqref{eq:ga} ensures $h(\mathbf{x}) \ge b(\mathbf{x},t)$
before $t_a$, so that, combined with the forward
invariance of the superlevel set of $b(\mathbf{x},t)$, the requirement
$G_{[t_a,t_b]}\mu$ is satisfied. When $c>0$ is small enough, the superlevel set of $b(\mathbf x,t)$ is always nonempty under Assumption \ref{as:feasible}. We choose the exponential function \eqref{eq:gamma_g} for \emph{always} instead of a linear function \eqref{eq:gamma_f} because it satisfies:
\begin{equation}
    0<-\gamma(t,\bm\omega)<c,\ \forall t\in[t_a,t_b].
\end{equation}
As a result, the condition $b(\mathbf{x},t)\ge 0$, equivalently 
$h(\mathbf{x})\ge -\gamma(t,\boldsymbol{\omega})$, is not overly conservative
on $[t_a,t_b]$ when $c>0$ is chosen sufficiently small.  
As detailed below, the HOCBF is deleted once $t>t_b$, which further reduces
conservativeness.

\noindent\textbf{Addressing Conflicts Between HOCBFs.}
We construct an HOCBF $b_j$ for each predicate $\mu_j$ in $\varphi$ using \eqref{eq:cat1} or \eqref{eq:cat23}.
However, the resulting constraints $b_j(\mathbf{x},t) \ge 0$ may conflict with each other during certain time intervals.
To resolve this issue, we introduce the following additional assumption:

\begin{assumption}
\label{as:circle}
For every predicate in Category II or III, the predicate function is of the
form
\begin{equation}
\label{eq:predicate}
    h(\mathbf{x})
    = s\big(R - \sigma (l(\mathbf{x}) - \mathbf{o})\big),
\end{equation}
where $s\in\{-1,1\}$, $R>0$, $\mathbf{o}\in\mathbb{R}^o$, $l:\mathbb{R}^n\!\to\!\mathbb{R}^o$
is a differentiable map (e.g., extracting a robot's position), and 
$\sigma :\mathbb{R}^o\!\to\!\mathbb{R}_{\ge0}$ is a smooth, strictly convex, radially unbounded gauge function, such as the Euclidean norm or a superelliptic norm.  
The corresponding sets $\mathcal{B}^{\mathrm{re}}(\mathbf{o},R)=\{\mathbf{z}:\sigma (\mathbf{z}-\mathbf{o})\le R\}, \mathcal{B}^{\mathrm{av}}(\mathbf{o},R)=\{\mathbf{z}:\sigma (\mathbf{z}-\mathbf{o})\ge R\},
$ represent generalized ``reach'' ($s=+1$) and ``avoid'' ($s=-1$) regions.
\end{assumption}

Predicates of the form \eqref{eq:predicate} describe reaching or avoiding a
smooth, strictly convex region (including circles, ellipses, superellipses,
and general convex gauge balls). Combined with temporal operators, they can
express rich high-level requirements. Circular or superelliptic regions are a
special case covered by Assumption~\ref{as:circle}. More general predicate
shapes will be investigated in future work.  Next, we give an example to illustrate the idea. 

\begin{example}
Consider $\varphi=F_{[0,2]}\mu_1\land F_{[2,4]}\mu_2$, where $
h_1(\mathbf{x}) = R_1 - \sigma (l(\mathbf{x}) - \mathbf{o}_1),
h_2(\mathbf{x}) = R_2 - \sigma (l(\mathbf{x}) - \mathbf{o}_2).
$
At time $t$, the invariant sets induced by their HOCBFs are
$\mathcal{B}^{\mathrm{re}}\big(\mathbf{o}_1,R_1+\gamma_1(t)\big),
\mathcal{B}^{\mathrm{re}}\big(\mathbf{o}_2,R_2+\gamma_2(t)\big).
$
To avoid conflicts between these two HOCBFs on the interval $[0,2]$, we:  
(1) enforce that their invariant sets have a nonempty intersection on $[0,2]$; and  
(2) delete $b_1$ once $h_1(\mathbf{x})>0$.  
A sufficient condition for (1) is $
\mathcal{B}^{\mathrm{re}}\left(\mathbf{o}_1,R_1+\gamma_1(2)\right)
\cap
\mathcal{B}^{\mathrm{re}}\left(\mathbf{o}_2,R_2+\gamma_2(2)\right)
\neq\emptyset.$
Using the triangle inequality of $\sigma$, this holds if
$\gamma_2(2)
\;\ge\;
\sigma (\mathbf{o}_1-\mathbf{o}_2)
- \gamma_1(2)
- R_1 - R_2.$
\end{example}

For each predicate wrapped with $G_{[t_a,t_b]}\mu$, we delete the associated
HOCBF at $t=t_b$.  
For $F_{[t_a,t_b]}\mu$, we delete the HOCBF once $h(\mathbf{x})\ge0$ for some
$t\ge t_a$.  
For $F_{[t_a,t_b]}(\bigwedge_{j=1}^N\mu_j)$, all corresponding HOCBFs are deleted
together once $h_j(\mathbf{x})\ge0$ for all $j$.  
Note that $G_{[t_a,t_b]}(\mu_1\land\mu_2)$ is equivalent to
$G_{[t_a,t_b]}\mu_1\land G_{[t_a,t_b]}\mu_2$, but
$F_{[t_a,t_b]}(\mu_1\land\mu_2)$ is not equivalent to
$F_{[t_a,t_b]}\mu_1\land F_{[t_a,t_b]}\mu_2$ since the latter allows
asynchronous satisfaction.

Next, reorder the predicates according to their ending times $t_b^1 \le t_b^2 \le \cdots \le t_b^M
$.
For predicate $\mu_j$ with ending time $t_b^j$ ($M\ge j\ge 2$), we impose $j-1$
nonconflict conditions in addition to \eqref{eq:f} or \eqref{eq:g}:
\begin{equation}
\label{eq:add_cons}
\begin{aligned}
\gamma_j(t_b^k)
\ge
s_j s_k \sigma (\mathbf{o}_j - \mathbf{o}_k)
- \gamma_k(t_b^k)
- s_k R_k - s_j R_j
+ D_{j,k}, \\
\qquad k=1,\ldots,j-1,
\end{aligned}
\end{equation}
where $\bm\omega$ is omitted, $h_k(\mathbf{x})
    = s_k\big(R_k - \sigma (l(\mathbf{x}) - \mathbf{o}_k)\big)$, $
D_{j,k}=\min(s_j+s_k,0)\times\inf,$
so that when $s_j=s_k=-1$ (both avoidance tasks), the condition is released.
For predicates in Category~I, we simply set $\gamma(t,\bm\omega)=0$.
Geometrically, the conditions in \eqref{eq:add_cons} ensure that, at each
predicate’s ending time, the feasible sets of all previously active HOCBFs
remain mutually compatible. To illustrate the idea, consider two predicates
$\mu_1$ and $\mu_2$ with ending times $t_b^1 \le t_b^2$. Since $b_1$ must be
maintained until $t_b^1$, we require that at $t_b^1$ the feasible region
induced by $b_2$ does not conflict with that of $b_1$, then the
corresponding compatibility condition becomes:
\begin{itemize}
\item \emph{reach--reach:} enforce 
      $\mathcal{B}^{\mathrm{re}}(\mathbf{o}_1,R_1+\gamma_1(t_b^1))\cap \mathcal{B}^{\mathrm{re}}(\mathbf{o}_2,R_2+\gamma_2(t_b^1))\neq\emptyset$;
\item \emph{reach--avoid:} enforce 
      $\mathcal{B}^{\mathrm{re}}(\mathbf{o}_1,R_1+\gamma_1(t_b^1))\cap \mathcal{B}^{\mathrm{av}}(\mathbf{o}_2,R_2-\gamma_2(t_b^1))\neq\emptyset$;
\item \emph{avoid--reach:} enforce 
      $\mathcal{B}^{\mathrm{av}}(\mathbf{o}_1,R_1-\gamma_1(t_b^1))\cap \mathcal{B}^{\mathrm{re}}(\mathbf{o}_2,R_2+\gamma_2(t_b^1))\neq\emptyset$;
\item \emph{avoid--avoid:} no additional condition is needed since 
      $\mathcal{B}^{\mathrm{av}}(\mathbf{o}_1,R_1-\gamma_1(t_b^1))\cap \mathcal{B}^{\mathrm{av}}(\mathbf{o}_2,R_2-\gamma_2(t_b^1))\neq\emptyset$ is always nonempty.
\end{itemize}
Because $\gamma(t,\bm\omega)$ in \eqref{eq:gamma_f} and \eqref{eq:gamma_g} is non-increasing,
the temporal operators are enforced as follows:
\begin{itemize}
\item For $F_{[t_a,t_b]}\mu$, constraint \eqref{eq:fb} guarantees the existence
      of some $t'\in[t_a,t_b]$ such that $h(\mathbf{x}(t'))\ge 0$.
\item For $G_{[t_a,t_b]}\mu$, constraint \eqref{eq:ga} ensures 
      $h(\mathbf{x}(t'))\ge 0$ for all $t'\in[t_a,t_b]$.
\item For $F_{[t_a,t_b]}(\bigwedge_{j=1}^N\mu_j)$, constraint \eqref{eq:fb}
      ensures the existence of a common $t'\in[t_a,t_b]$ where all
      $h_j(\mathbf{x}(t'))\ge 0$.
\end{itemize}

\begin{theorem}
\label{thm:construct}
Consider system~\eqref{eq:system1} satisfying 
Assumptions~\ref{as:feasible} and \ref{as:circle}, and HOCBFs constructed by
\eqref{eq:cat1} and \eqref{eq:cat23} that satisfy 
\eqref{eq:f}, \eqref{eq:g}, and \eqref{eq:add_cons}.
Suppose each HOCBF satisfies $\psi_{i,j}(\mathbf{x}_0,0)\ge 0$ for 
$i=0,\ldots,m-1$, $j=1,\ldots,M$. Then any control input $\mathbf{u}(t)$ that satisfies
\eqref{eq: hocbf} for all relevant HOCBFs guarantees satisfaction of the STL
specification $\varphi$.
\end{theorem}

\begin{proof}
Constraints \eqref{eq:f0} and \eqref{eq:g0} ensure $b_j(\mathbf{x}_0,0)\ge 0$
for all predicates. Since $\psi_{i,j}(\mathbf{x}_0,0,\boldsymbol{\theta}_b)\ge 0$ for all $i$, 
Theorem~\ref{thm:hocbf} guarantees that any control input satisfying
\eqref{eq: hocbf} yields $b_j(\mathbf{x},t)\ge 0$ for all undeleted HOCBFs.  
Because $\gamma_{j}(t,\bm\omega_j)$ is non-increasing, \eqref{eq:fb} ensures $\exists t'\in[t_a,t_b]$, $h_{j}(\mathbf x(t'))\geq0$ for $F{[t_a,t_b]}\mu$, while \eqref{eq:ga} ensures $\forall t'\in[t_a,t_b]$, $h_{j}(\mathbf x(t'))\geq0$ for $G{[t_a,t_b]}\mu$. For a formula in the form of $F_{[t_a,t_b]}(\bigwedge_{j=1}^N\mu_j)$, \eqref{eq:fb} ensures $\exists t'\in[t_a,t_b]$, $h_j(\mathbf x(t'))\geq0$ for all $j$. Hence every predicate in $\varphi$ is
satisfied over its required temporal interval, ensuring $(\mathbf{x},0)\models
\varphi$.
\end{proof}

\subsection{Unified Robustness Metric}
\label{subsec:Unified Rob}
In Sec.~\ref{subsec:hocbf-stl}  we constructed either  
(i) a time-invariant HOCBF $b_j(\mathbf{x})$ (for Category~I predicates),  
or  
(ii) a time-varying HOCBF $b_j(\mathbf{x},t,\bm\theta_b)$ (for Category~II and III predicates),  
together with a deletion rule that removes the corresponding constraint once
the predicate is satisfied (for $F_{[t_a,t_b]}\mu$ operators) or at the terminal time
(for $G_{[t_a,t_b]}\mu$ operators).
Before deletion, each HOCBF imposes a highest-order constraint
$\psi_{m,j}(\mathbf{x},t,\boldsymbol{\theta}_b,\boldsymbol{\theta}_p)\ge0$, which must remain feasible for the dQP \eqref{eq:barriernet}.
This subsection introduces a unified robustness combining:
(i) feasibility robustness, and
(ii) STL temporal robustness based on the exponential robustness in
\cite{liu2023robust}, ensuring that the controller remains feasible while
satisfying the STL specification.

\noindent\textbf{Exponential robustness for STL predicates.}
For a predicate $\mu: h(\mathbf{x})\ge0$, the instantaneous robustness is
\begin{equation}
\label{eq:pho1}
\rho(\mu,\mathbf{x},t)=h(\mathbf{x}(t)),\qquad
\rho(\neg\mu,\mathbf{x},t)=-h(\mathbf{x}(t)).
\end{equation}
Because ``always'' and ``eventually'' can be expressed using conjunction
(and disjunction eliminated via De Morgan's law), the exponential robustness
is defined for conjunction over $M$ subformulas with robustness values
$\rho_1,\ldots,\rho_M$.
Let $\rho_{\min}=\min(\rho_1,\ldots,\rho_M)$ and define
\begin{equation}
\label{eq:pho2}
\rho_j^{\mathrm{conj}}
\coloneqq
\begin{cases}
\rho_{\min}e^{\frac{\rho_j-\rho_{\min}}{\rho_{\min}}}, & \rho_{\min}<0,\\[1mm]
\rho_{\min}(2-e^{\frac{\rho_{\min}-\rho_j}{\rho_{\min}}}), & \rho_{\min}>0,\\[1mm]
0, & \rho_{\min}=0,
\end{cases}
\end{equation}
and the exponential robustness for conjunction:
\begin{equation}
\label{eq:pho3}
\mathcal{A}^{\exp}(\rho_1,\ldots,\rho_M)
=
\beta\rho_{\min}
+
(1-\beta)
\frac{1}{M}\sum_{j=1}^M \rho_j^{\mathrm{conj}},
\end{equation}
where $\beta\in[0,1]$ balances the minimum and the average.  This robustness
is \emph{sound}: $\rho(\varphi,\mathbf{x},0)\ge0$ iff
$(\mathbf{x},0)\models\varphi$.

\noindent\textbf{Feasibility robustness.}
As shown in Theorem~\ref{thm:hierarchical-feasibility}, the dQP is feasible at
time $t$ whenever $\psi_{m-1,j}(\mathbf{x},t,\boldsymbol{\theta}_b,\boldsymbol{\theta}_p)>0$.  Thus the instantaneous
feasibility margin is
\begin{equation}
\label{eq:pho4}
\rho^{\mathrm{fea}}_j(\mu_j,\mathbf{x},t)
:= \psi_{m-1,j}(\mathbf{x},t,\boldsymbol{\theta}_b,\boldsymbol{\theta}_p).
\end{equation}
This margin must stay positive \emph{until} the HOCBF is deleted. An HOCBF corresponding to a predicate $\mu_j$ enclosed by $F_{[t_a^j,t_b^j]}$ is deleted once $h_j(\mathbf{x})\ge0$. We use an additional STL subformula to express the feasibility requirement:
\begin{equation}
\label{eq:feasibility-f}
\varphi_j^{\mathrm{fea}}:=
\big(\psi_{m-1,j}(\mathbf{x},t,\boldsymbol{\theta}_b,\boldsymbol{\theta}_p) > 0\big) U_{[t_a^j,t_b^j]} \big(h_j(\mathbf{x})\ge0\big),
\end{equation}
where $U_{[t_a^j,t_b^j]}$ is the temporal \emph{Until} operator as defined in standard STL. An HOCBF corresponding to a predicate $\mu_j$ enclosed by $G_{[t_a^j,t_b^j]}$ is always deleted at $t_b^j$. Similiarly, we define the feasibility STL subformula as:
\begin{equation}
\label{eq:feasibility-g}
\varphi_j^{\mathrm{fea}}:=
G_{[t_a^j,t_b^j]}\big(\psi_{m-1,j}(\mathbf{x},t,\boldsymbol{\theta}_b,\boldsymbol{\theta}_p) > 0 \big).
\end{equation}
With the construction above, we now have $M$ predicates
$\mu_j$, $j=1,\ldots,M$, spanning Categories~I--III.
Each predicate $\mu_j$ is associated with an STL subformula
$\varphi_j^{\mathrm{fea}}$ that enforces QP feasibility until deletion,
given by expression \eqref{eq:feasibility-f} and \eqref{eq:feasibility-g}.

To capture the control bounds \eqref{subeq:input-bound}, we introduce another STL formula defined over the control signal $\mathbf u$ which must hold for all $t\in[0,T]$:
\begin{equation}
\label{eq:pho-u}
\varphi^{\mathcal U}
:= 
G_{[0,T]}\big((\mathbf{u}(t)-\mathbf{u}_{\min}\ge \mathbf{0})
\wedge
(\mathbf{u}_{\max}-\mathbf{u}(t)\ge \mathbf{0})\big).
\end{equation}
Note that $\varphi_j^{fea}$ and $\varphi^{\mathcal U}$ in \eqref{eq:feasibility-f}, \eqref{eq:feasibility-g} and \eqref{eq:pho-u} are additional STL formulas other than the original STL specification defined in \eqref{eq:stl}, so they have different syntax. 

\noindent\textbf{Feasibility-aware STL robustness.}
We define an augmented STL specification that combines the original
task $\varphi$ with all feasibility conditions and control bounds:
\begin{equation}
\label{eq:pho7}
\widetilde{\varphi}:=
\varphi \wedge \varphi^{\mathcal U} \wedge
\bigwedge_{j=1}^M \varphi_j^{\mathrm{fea}}.
\end{equation}
The robustness of $\widetilde{\varphi}$ can be evaluated using the exponential
robustness in~\eqref{eq:pho1}--\eqref{eq:pho3}.  
We denote the resulting unified scalar robustness by
\begin{equation}
\label{eq:pho8}
\rho^{\mathrm{uni}}(\mathbf{x},0)
\;:=\;
\rho(\widetilde{\varphi},\mathbf{x},0),
\end{equation}
where the conjunction in~\eqref{eq:pho7} is handled via
$\mathcal{A}^{\exp}$ in~\eqref{eq:pho3} and the effective robustness
$\rho_j^{\mathrm{conj}}$ in~\eqref{eq:pho2}.
\begin{theorem}[Feasibility-Aware Correctness of BarrierNet Controllers]
\label{thm:fea-stl}
Consider the system~\eqref{eq:system1} and the STL specification $\varphi$ satisfying Assumptions~\ref{as:feasible} and \ref{as:circle}. Construct the HOCBFs and their feasibility subformulas $\varphi_j^{\mathrm{fea}}$ as in \eqref{eq:feasibility-f}, \eqref{eq:feasibility-g}, and define the augmented STL specification $\widetilde{\varphi}$ by combining $\varphi$ with all feasibility and control-bound subformulas as in~\eqref{eq:pho7}. Let a neural network controller be given whose last layer is the differentiable QP~\eqref{eq:barriernet}, which enforces all highest-order HOCBF constraints, deletion rules, and control bounds at all time.  
If the dQP~\eqref{eq:barriernet} remains feasible for all $t>0$, then the closed-loop trajectory generated by this controller guarantees the satisfaction of the original STL specification $\varphi$. 
Moreover, if each STL subformula
$\varphi_j^{\mathrm{fea}}$ is satisfied, i.e., $(\mathbf x,0)\models\bigwedge_{j=1}^M \varphi_j^{\mathrm{fea}}$, then the dQP~\eqref{eq:barriernet} remains feasible for all $t\ge0$.
\end{theorem}

\begin{proof}
By Theorems \ref{thm:dqp-hocbf} and \ref{thm:construct}, the controller satisfying \eqref{eq:barriernet} guarantees the resulting closed-loop trajectory satisfies the STL specification $\varphi$.
By satisfaction of each feasibility subformula $\varphi_j^{\mathrm{fea}}$, we have that $\psi_{m-1,j}(\mathbf{x},t,\boldsymbol{\theta}_b,\boldsymbol{\theta}_p)>0$ for all $t$ prior to deletion. By Theorem~\ref{thm:hierarchical-feasibility}, the dQP remains pointwise feasible at all times. 
\end{proof}
\begin{corollary}\label{cor:rho_uni_short}
If the conditions of Theorem~\ref{thm:fea-stl} hold, then
$\rho^{\mathrm{uni}}(\mathbf{x},0):=\rho(\widetilde{\varphi},\mathbf{x},0)\ge0$.
\end{corollary}
\begin{proof}
By Theorem~\ref{thm:fea-stl}, the stated controller conditions imply
$(\mathbf{x},0)\models \widetilde{\varphi}$. By soundness of exponential robustness,
$(\mathbf{x},0)\models \widetilde{\varphi}$ implies $\rho(\widetilde{\varphi},\mathbf{x},0)\ge0$,
i.e., $\rho^{\mathrm{uni}}(\mathbf{x},0)\ge0$.
\end{proof}

\begin{remark}
\label{rem:rem2}
 The original STL specification $\varphi$ is theoretically guaranteed in continuous time through HOCBF constraints embedded in the dQP formulation. However, in practice, the resulting QP is implemented in discrete time, which introduces inter-sampling effects \cite{liu2025sampling} and may affect both the continuous-time satisfaction of the specification and its discrete-time evaluation at the sampling instants. Therefore, the robustness of the original STL specification $\varphi$ in \eqref{eq:pho7} is monitored to reflect its satisfaction under discrete-time enforcement. The feasibility-related STL subformulas $\bigwedge_{j=1}^M \varphi_j^{\mathrm{fea}}$ are not enforced by hard constraints and thus are not strictly guaranteed. Instead, their robustness values, together with the satisfaction of the input-bound specification $\varphi^{\mathcal U}$, are incorporated as optimization objectives during training to guide and improve feasibility. Addressing inter-sampling effects is left for future work.
\end{remark}

\subsection{Learning Robust Controllers}
\label{subsec:learning}

Corollary \ref{cor:rho_uni_short} shows that maintaining a feasibility-aware STL robustness above zero indicates both (i) satisfaction of the STL specification and (ii) feasibility of the underlying QP at all times.
With this monitoring, the BarrierNet architecture in~\eqref{eq:barriernet}, integrated with a
model-based RL framework, can be used to synthesize a controller that
accounts for the feasibility-aware robustness during execution while 
gradually improving it through training. In this subsection, we first explain why in \eqref{eq:barriernet}, $b(\mathbf x,t,\bm\theta_b,\mathbf x_0)$ and $\mathbf{P}_{\text{inip}}(\mathbf{x}_0,\boldsymbol{\theta}_{\text{inip}})$ all depend on the initial condition $\mathbf x_0$. Then we describe the structure of the entire neural network controller $\pi(\mathbf x(t),\bm\theta)$ in \eqref{eq:nn-input} or $\pi(\mathbf x_{0:t},\bm\theta)$ in \eqref{eq:nn-input2}. Finally, we introduce the training process of the controller. 

\noindent\textbf{Hyperparameters Depending on the Initial Condition.}
Consider a predicate in Category~II or III with the associated HOCBF
$b(\mathbf{x},t,\bm\theta_b)=h(\mathbf{x})+\gamma(t,\boldsymbol{\omega})$.
The constraints \eqref{eq:f0} and \eqref{eq:g0} that shape the function
$\gamma(t,\boldsymbol{\omega})$ explicitly depend on the initial state
$\mathbf{x}_0$.  
Therefore, the hyperparameter vector $\boldsymbol{\omega}$ must be adapted to
each initial condition rather than fixed globally.

To accomplish this, we generate $\boldsymbol{\omega}$ using a neural
network that takes $\mathbf{x}_0$ as input:
$\boldsymbol{\omega}=\boldsymbol{\omega}(\mathbf{x}_0,\boldsymbol{\theta}_b)$.
Consequently, the resulting HOCBF also depends on the initial condition and
carries trainable parameters $\boldsymbol{\theta}_b$, which we denote by
$b(\mathbf{x},t,\boldsymbol{\theta}_b,\mathbf{x}_0)$.

To ensure that the HOCBF constraint hierarchy enforces set invariance, we must also
satisfy $\psi_i(\mathbf{x}_0,0,\boldsymbol{\theta}_b,\boldsymbol{\theta}_p)\ge 0$ for all $i=1,\ldots,m-1$.  
Since $b(\mathbf{x}_0,0,\bm\theta_b)>0$, \eqref{eq: psi2} implies that each
$\psi_i(\mathbf{x}_0,0,\boldsymbol{\theta}_b,\boldsymbol{\theta}_p)$ can be made nonnegative by selecting sufficiently
large multipliers $p_{i}^{\mathrm{inip}}$ at $t=0$. To ensure the satisfaction of the initial feasibility condition, the parameters $p_{i}^{\mathrm{inip}}$ are selected such that each $\psi_{m-1,j}(\mathbf{x}, t, \boldsymbol{\theta}_b, \boldsymbol{\theta}_p) = \epsilon_{j}>0$, where $\epsilon_{j}$ is a sufficiently small positive constant. 
Because these lower bounds on $p_{i}^{\mathrm{inip}}$ depend on the initial condition
$\mathbf{x}_0$, we generate them through a neural network that takes
$\mathbf{x}_0$ as input and is parameterized by
$\boldsymbol{\theta}_{\mathrm{inip}}$, then we have
$p_{i}^{\mathrm{inip}} = p_{i}^{\mathrm{inip}}(\mathbf{x}_0,\boldsymbol{\theta}_{\mathrm{inip}})$.
We denote the resulting initial multipliers collectively as
$\mathbf{p}_{j}= [p_{1,j}^{\mathrm{inip}},\ldots,p_{m-1,j}^{\mathrm{inip}}]^{\top}$ and $\mathbf{P}_{\text{inip}}=[\mathbf{p}_1^{\top},\ldots,\mathbf{p}_M^{\top}]\in\mathbb{R}^{(m-1)M}$. $\mathbf{P}_{\mathrm{inip}}$ serves as the initialization of
$p_{i,j}(\mathbf{x},\boldsymbol{\theta}_p,\mathbf{P}_{\mathrm{inip}})$ at $t=0$. 

\noindent\textbf{InitNet Structure.}
In practice, we use a single neural network, referred to as \emph{InitNet}, to
generate all parameters that depend on the initial condition $\mathbf{x}_0$:
\begin{equation}
    \label{eq:nn_init}
    [\bm\Omega\;\mathbf{P}_{\text{inip}}]
    = \mathfrak{N}(\mathbf{x}_0,\bm\theta_b,\boldsymbol{\theta}_{\mathrm{inip}}),
\end{equation}
where
$\bm\Omega=[\bm\omega_1^{\top},\ldots,\bm\omega_N^{\top}]\in\mathbb{R}^{2N},N\le M$ 
collects all hyperparameters defining the time-varying terms $\gamma_j(t,\bm\omega_j)$.
The network $\mathfrak{N}$ is parameterized by the trainable variables
$\bm\theta_b$ and $\boldsymbol{\theta}_{\mathrm{inip}}$.
All constraints on $\gamma_{j}$ in \eqref{eq:f}, \eqref{eq:g}, and
\eqref{eq:add_cons} are equivalently enforced as constraints on the output
$\bm\Omega$ of InitNet. For hyperparameter constraints of the form 
$\omega\in[\underline{\omega},\overline{\omega}]$, 
we enforce them by applying a \emph{Sigmoid} activation at the final layer of 
$\mathfrak N$. 
For one-sided constraints such as 
$\omega\in[\underline{\omega},\infty)$ or $\omega\in(-\infty,\overline{\omega}]$, 
we use a \emph{Softplus} activation (possibly combined with sign flipping) to 
guarantee positivity or negativity as required.  
With these activations, the generated vector $\bm\Omega$ automatically satisfies 
all constraints in~\eqref{eq:f},~\eqref{eq:g}, and~\eqref{eq:add_cons}. Similarly, for the multipliers $p_{i}^{\mathrm{inip}}$, $i=1,\ldots,m-1$, 
the initialization constraints
\begin{equation}
    \label{eq:p_init}
p_{i}^{\mathrm{inip}} > \max\!\Big\{
-\dot\psi_{i-1}(\mathbf{x}_0,0,\boldsymbol{\omega},\mathbf{p})\big/\alpha_i(\psi_{i-1}(\mathbf{x}_0,0,\boldsymbol{\omega},\mathbf{p})),\; 0
\Big\}
\end{equation}
are enforced by applying \emph{Softplus} activations to the corresponding 
output channels of $\mathfrak N$, ensuring that each $p_{i}^{\mathrm{inip}}$ satisfies its 
required lower bound. InitNet is only used at time $t=0$ to provide a set of HOCBFs and the initial multipliers for corresponding class $\kappa$ functions, which are then used to train $\bm\theta_b$ to obtain $p_{i,j}(\mathbf{x},\boldsymbol{\theta}_p,\mathbf{P}_{\mathrm{inip}})$.

\noindent\textbf{RefNet Structure.}
At each discrete time step, we employ a neural network (with or without memory), referred to
as \emph{RefNet}, parameterized by $(\boldsymbol{\theta}_f,\boldsymbol{\theta}_q)$---to generate the reference terms $
\mathbf{F}(\mathbf{x}_{0:t},\boldsymbol{\theta}_f),
\mathbf{Q}(\mathbf{x}_{0:t},\boldsymbol{\theta}_q)
$ for controller $\pi(\mathbf x_{0:t},\bm\theta)$ in \eqref{eq:nn-input2} or $
\mathbf{F}(\mathbf{x}(t),\boldsymbol{\theta}_f),
\mathbf{Q}(\mathbf{x}(t),\boldsymbol{\theta}_q)$ for controller $\pi(\mathbf x(t),\bm\theta)$ in \eqref{eq:nn-input}---which serve as reference inputs to the BarrierNet module.

\noindent\textbf{BarrierNet Structure.} In this paper, BarrierNet is implemented as a neural network that takes the system state $\mathbf{x}(t)$ as input and outputs the multipliers $p_{i,j}(\mathbf{x},\boldsymbol{\theta}_p,\mathbf{P}_{\mathrm{inip}})$ and the optimal control input $\mathbf{u}^*(t)$, where the network is parameterized by $\boldsymbol{\theta}_p$. The hyperparameters of the time-varying HOCBFs are obtained from the InitNet
$\mathfrak{N}(\mathbf{x}_0,\boldsymbol{\theta}_b,\boldsymbol{\theta}_{\mathrm{inip}})$, which
produces the vectors $\boldsymbol{\Omega}$ and $\mathbf{P}_{\mathrm{init}}$ that
encode all $\gamma$-function hyperparameters and the initial multipliers.

Given the outputs of RefNet and InitNet, the BarrierNet layer solves the differentiable QP in~\eqref{eq:barriernet} at each time step to produce the multipliers $p_{i,j}(\mathbf{x},\boldsymbol{\theta}_p,\mathbf{P}_{\mathrm{inip}})$ and the control input $\mathbf{u}^*(t)$. The control input $\mathbf{u}^*(t)$ is then applied sequentially to the system~\eqref{eq:system1}, allowing BarrierNet to roll out the closed-loop trajectory. Based on the resulting trajectory $\mathbf{x}_{0:t_{T}}$ and the corresponding multipliers $p_{i,j}(\mathbf{x},\boldsymbol{\theta}_p,\mathbf{P}_{\mathrm{inip}})$, the overall robustness $\rho(\tilde{\varphi},\mathbf{x}_{0:t_{T}},0)$ of the augmented STL specification, together with the associated cost $J(\mathbf{u}_{0:t_{T}-1})$, is evaluated.
This trajectory-level evaluation is then used to update the controller
parameters $\boldsymbol{\theta}$ during reinforcement training.
Overall, the neural network controller
$
\pi(\cdot,\boldsymbol{\theta})$, where 
$\boldsymbol{\theta}
=(\boldsymbol{\theta}_q,\boldsymbol{\theta}_f,\boldsymbol{\theta}_b,\boldsymbol{\theta}_{\text{inip}},\boldsymbol{\theta}_p),$
consists of  
(i) InitNet for initial hyperparameter generation,
(ii) RefNet for predicting $\mathbf{F}$ and $\mathbf{Q}$, and  
(iii) the BarrierNet dQP layer that enforces the HOCBF constraints.
The full architecture is illustrated in Fig.~\ref{fig:framework}.

\noindent\textbf{Training Neural Network Controller.}
Following the procedure in \cite{liu2023safe}, we randomly sample
$V$ initial conditions $\mathbf{x}_0^v$, $v=1,\ldots,V$.  
For each initial state, we roll out the closed-loop system
\eqref{eq:system1} under the controller $\pi$ up to the horizon $T$,
thereby generating $V$ state and control trajectories.  
For each trajectory, we compute the unified STL robustness and its associated
cost $J$, and use their empirical mean to approximate the expectation.
Thus, the optimization problem in~\eqref{eq:goal} can be written as
\begin{equation}
    \label{eq:cost}
    \begin{aligned}
       \bm\theta^* &= \arg\max_{\bm\theta} \frac{1}{V}\sum_{v=1}^V\big[\rho(\widetilde{\varphi},\mathbf x_{0:t_{T}}^v, 0) - J(\mathbf u_{0:t_{T}-1}^v)\big] \\
       \text{s.t.}\ & \dot{\mathbf x}^v = f(\mathbf{x}^v) + g(\mathbf{x}^v)\mathbf{u}(t),\ v=1,\ldots,V,\\
       &\mathbf u^v(t)=\pi(\mathbf x^v(t),\bm\theta)\ \text{or} \ \mathbf u^v(t)=\pi(\mathbf x^v_{0:t},\bm\theta),
    \end{aligned}
\end{equation}
where the superscript $v$ denotes the $v$-th sample.  
By substituting the system dynamics into the objective, the problem becomes an
unconstrained optimization over the parameters of the controller.  
Since the QP in \eqref{eq:barriernet} is differentiable with respect to its
inputs and parameters via the OptNet technique \cite{amos2017optnet}, we
backpropagate the gradient of the objective in \eqref{eq:cost} through the QP
layer to all components of $\bm\theta$.  
The gradients of the STL robustness are computed analytically and
automatically using a modified implementation of STLCG
\cite{leung2020backpropagation} adapted to the exponential robustness in
\cite{liu2023robust}.  
We then update the parameters using these gradients.  
At each gradient step, we resample $V$ initial conditions from the set
$\mathcal{X}_0$ to improve exploration of the initial-state distribution, and
we employ the stochastic optimizer Adam \cite{kingma2014adam} to perform the
training. We summarize our solution to Problem \ref{pb:1} in Algorithm \ref{alg:1}.

\begin{algorithm}
\KwIn{System dynamics \eqref{eq:system1}, control bounds \eqref{subeq:input-bound}, horizon $T$ and STL formula $\varphi$}
\KwOut{Robust and correct controller $\pi(\cdot,\bm\theta^*)$}
Construct HOCBFs from $\varphi$ using \eqref{eq:cat1}, \eqref{eq:cat23}\;
Set up constraints on $\bm\Omega$ using \eqref{eq:f}, \eqref{eq:g}, \eqref{eq:add_cons}\;
Set up constraints on $\mathbf{P}_{\mathrm{inip}}$ using \eqref{eq: psi2}\;
Initialize controller $\pi(\cdot,\bm\theta)$ including RefNet $\mathbf Q(\cdot,\bm\theta_q)$, $\mathbf{F}(\cdot,\bm\theta_f)$, InitNet \eqref{eq:nn_init} and BarrierNet (the dQP) \eqref{eq:barriernet}\;
Construct the augmented STL formula $\widetilde{\varphi}$ \eqref{eq:pho7}\;
\Repeat{Convergence; \Return $\bm\theta^*$}{
Sample $V$ initial conditions $x_0^v$\;
Obtain $\bm\Omega$, $\mathbf{P}_{\text{inip}}$ for each $x_0^v$ from InitNet \eqref{eq:nn_init}\;
Evaluate \eqref{eq:cost} by applying $\pi(\cdot,\bm\theta)$ to \eqref{eq:system1}\;
Compute gradient of \eqref{eq:cost} w.r.t. $\bm\theta$\;
Update $\bm\theta$ using Adam optimizer\;
}
 \caption{Construction and training of controller}\label{alg:1}
\end{algorithm}

\section{CASE STUDIES}
\label{sec:sim}
In this section, we demonstrate the efficacy of our approach through simulations on a 2D robot navigation problem and compare it with existing algorithms.
\subsection{Case Study I: Linear Dynamical System — Double Integrator}
\label{subsec:case 1}
Consider a robot with double integrator model in the form:
\begin{equation}
    \label{eq:robot}
    \begin{bmatrix} \dot{x}\\ \dot{y}\\ \dot{v}_x\\ \dot{v}_y\end{bmatrix} = \begin{bmatrix} v_x\\ v_y\\ 0\\ 0\end{bmatrix} + \begin{bmatrix} 0&0\\ 0&0\\ 1&0\\ 0&1\end{bmatrix} \begin{bmatrix}u_1\\u_2\end{bmatrix},
\end{equation}
where $\mathbf{x} = [x\ y\ v_x\ v_y]^\top$ and $\mathbf{u} = [u_1\ u_2]^\top$, with $[x\ y]^\top$ denoting the 2D position, $[v_x\ v_y]^\top$ the velocity, and $[u_1\ u_2]^\top$ the acceleration of the robot.
During the initialization of the neural network controller, control input bounds are enforced to ensure a proper warm start. However, in the subsequent optimization stage, the inputs are treated as unconstrained, while the STL specification still accounts for the actual input limits in \eqref{eq:pho-u}, and the cost function in \eqref{eq:cost} employs an $\ell_2$ penalty with a weight of $0.003$ to discourage large accelerations.
\subsubsection{Neural Network Controller without Memory}
\label{subsubsec:caseI1}
Consider the environment shown in Fig. \ref{fig:env}. $\mathbf x_0$ is uniformly sampled in the region $Init$ with zero velocity. We discretize the system with a time interval of $0.1s$. The task for the robot is given by an STL formula:
\begin{equation}
    \label{eq:task}
    \varphi = F_{[0,2]}Reg_1 \land F_{[2,5]}Reg_2 \land G_{[0,5]}(\neg Obs_3 \land \neg Obs_4),
\end{equation}
where $Reg_j$ indicates $R_j-\|l(\mathbf x) - \mathbf o_j\|_2\geq 0$, $j=1,2$, $l(\mathbf x) = [x\ y]^\top$. The control bounds given by an STL formula is
\begin{equation}
\label{eq:pho-u2}
\varphi^{\mathcal U}
= 
G_{[0,5]}\big((\mathbf{u}-\mathbf{u}_{\min}\ge \mathbf{0})
\wedge
(\mathbf{u}_{\max}-\mathbf{u}\ge \mathbf{0})\big),
\end{equation}
where $\mathbf{u}_{\min}=[-10\ -10]^{\top},$ and $\mathbf{u}_{\max}=[10\ 10]^{\top}.$ $Obs_j$ is a superellipse:
\begin{equation}
    \label{eq:obs}
    1-\sqrt[4]{(\frac{x-o_{x,i}}{a_i})^4 + (\frac{y-o_{y,i}}{b_i})^4}\geq 0,\ j=3,4.
\end{equation}
Here, $Reg_j$ belongs to Category II and $Obs_j$ belongs to Category I for $j=1,2,3,4$.
In plain terms, the STL formula $\varphi$ requires the robot to reach $Reg_1$ within the time interval $[0,2]$ and $Reg_2$ within $[2,5]$, while avoiding obstacles $Obs_3$ and $Obs_4$ at all times.
The time horizon of $\varphi$ is $5$.
All four predicates have a relative degree of $2$ with respect to system~\eqref{eq:robot}. In this example, we set $\mathbf Q(\mathbf x(t), \bm\theta_q)$ to be an identity matrix, so for all $t>0$ the output of the previous layers reduces to $\mathbf F(\mathbf x(t), \bm\theta_f)$, which serves as a reference control.
Because the task does not require back-and-forth motions, a memory-based structure is unnecessary for generating $\mathbf F$.
Accordingly, RefNet, InitNet, and BarrierNet are implemented as feedforward neural networks with three fully connected layers each.
For the robustness measure, we adopt the exponential robustness in \eqref{eq:pho2} and \eqref{eq:pho3}.

\begin{figure}[t]
    \centering
    \begin{subfigure}[b]{0.22\textwidth}
        \centering        \includegraphics[width=\textwidth]{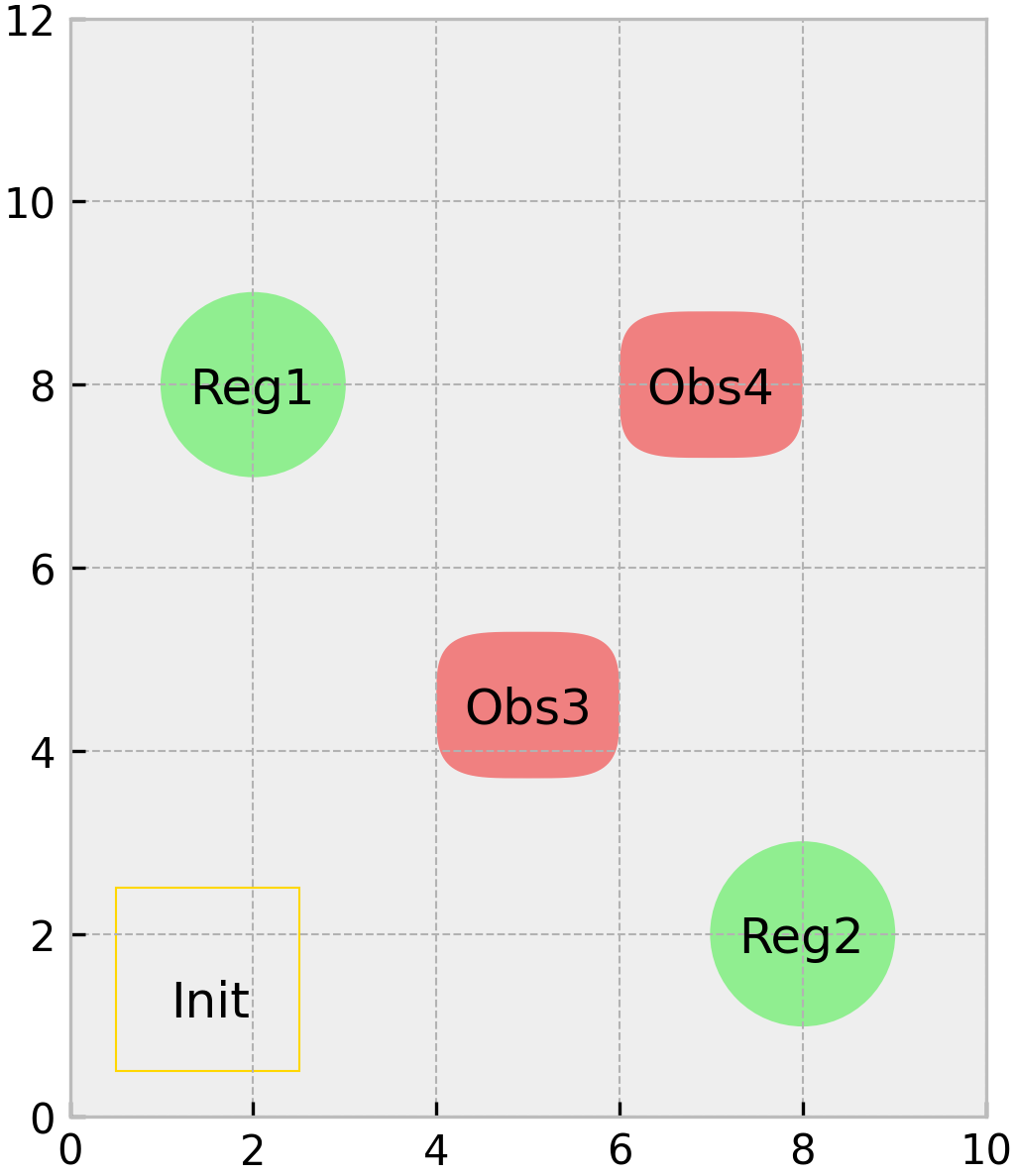}
        \caption{\small Environment}
        \label{fig:env}
    \end{subfigure}
    \hfill
    \begin{subfigure}[b]{0.26\textwidth}
        \centering        \includegraphics[width=\textwidth]{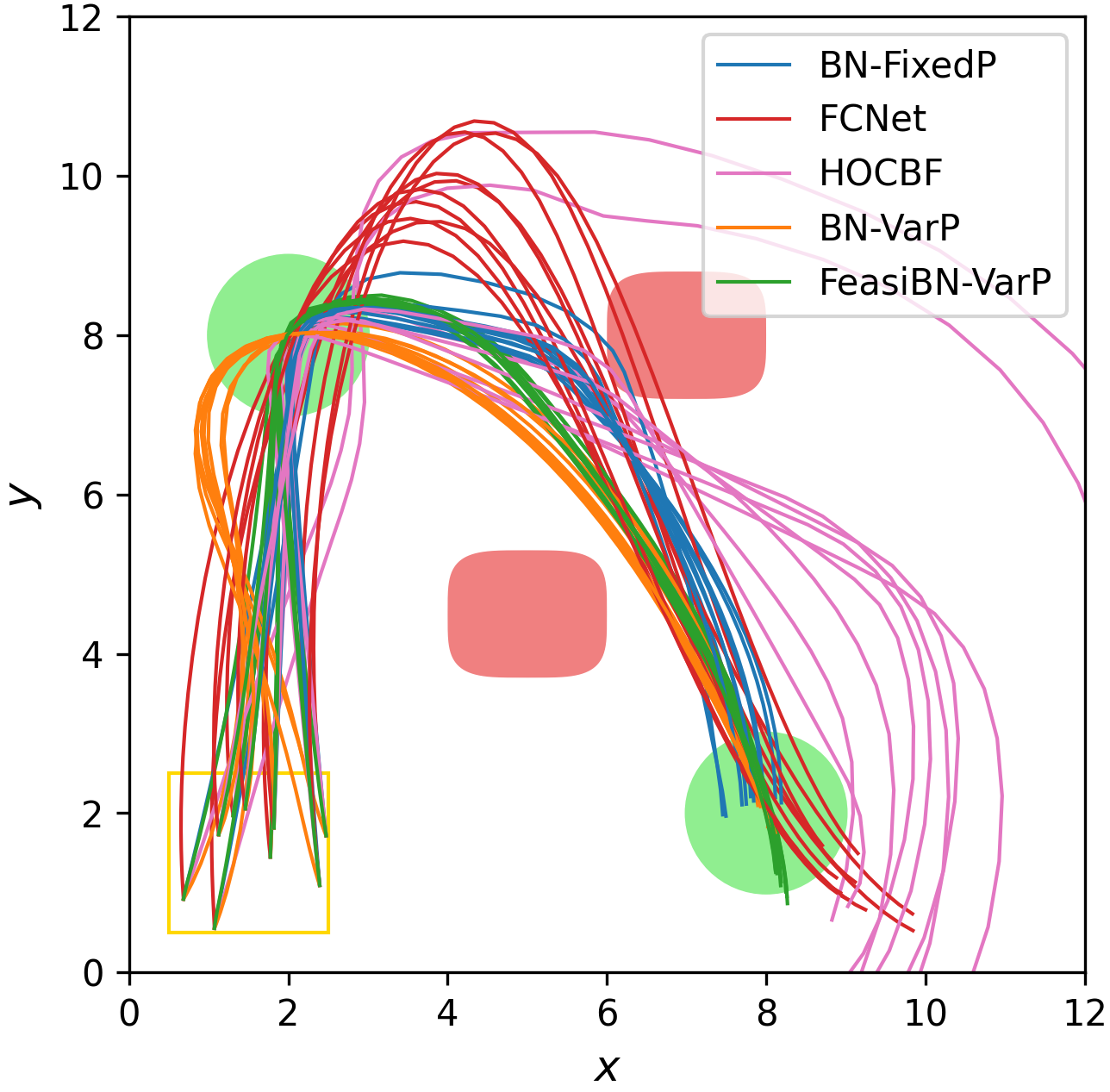}
        \caption{\small Trajectories}
        \label{fig:traj}
    \end{subfigure}
    \caption{(a) The 2D environment. (b) 10 trajectories with sampled initial conditions using BarrierNet (fixed multipliers (blue), time-varying multipliers (orange and green), feasibility-aware (green)), FCNet (red), and HOCBF (pink).}
    \label{fig:trajectories}
\end{figure}
\noindent\textbf{Comparison setup.} We construct the HOCBFs \eqref{eq:cat1}, \eqref{eq:cat23} and train the controller proposed in this paper under two different configurations.
(i) In the first configuration, the multipliers $p_{i,j}$ are time-varying, while the feasibility subformulas \eqref{eq:feasibility-f}, \eqref{eq:feasibility-g} are not incorporated; we refer to this controller as BN-VarP.
(ii) In the second configuration, the multipliers $p_{i,j}$ are time-varying and the feasibility subformulas \eqref{eq:feasibility-f}, \eqref{eq:feasibility-g} are incorporated; we refer to this controller as FeasiBN-VarP. The position state $[x\ y]$ serves as the input to the trained multiplier $p_{i,j}([x\ y], \boldsymbol{\theta}_p, \mathbf{P}_{\text{init}} )$.

For comparison, we include our previous work \cite{liu2023learning}, where the multipliers $p_{i,j}$ remain constant within each rollout; we call this controller BN-FixedP. Another benchmark is \cite{liu2023safe}, in which a neural network controller without BarrierNet—i.e., a standard Fully Connected Neural Network (FCNet)—is trained to satisfy an STL task. This is equivalent to directly using the reference control $\mathbf F$. We refer to this controller as FCNet. To ensure a fair comparison, we assume that the system dynamics are known for \cite{liu2023safe}, and we use the same objective function, optimizer, and neural network architectures; that is, BN-FixedP and FCNet share the same architecture as $\mathbf F(\mathbf x, \bm\theta_f)$.

The training curves for all learning-based methods are shown in Fig.~\ref{fig:curves}.
In addition, we apply the method in \cite{lindemann2018control} (extended to the HOCBF setting) without learning: we construct HOCBFs with fixed hyperparameters and solve the QP \eqref{eq:barriernet} with $\mathbf F = \mathbf 0$. The hyperparameters are randomly chosen but satisfy all constraints \eqref{eq:f}, \eqref{eq:g}, and \eqref{eq:add_cons}. We refer to this baseline as HOCBF. The resulting average objective values and robustness (evaluated only for the original task $\varphi$ and the control-bound specification $\varphi^{\mathcal U}$) over 10 random initial conditions are shown in Fig.~\ref{fig:curves}, where the HOCBF baseline is illustrated by pink dashed lines. Representative trajectories obtained by the five approaches under 10 random initial conditions are shown in Fig.~\ref{fig:traj}. The inputs over time for the five methods are shown in Fig.~\ref{fig:inputs}. We select the second task, $F_{[2,5]}Reg_2$ ($j=2$), and the fourth task, $G_{[0,5]}(\neg Obs_4)$ ($j=4$), to illustrate the corresponding $p_{i,j}$ values ($i=1,2$) for the five methods in Fig. \ref{fig:multipliers}.

\begin{figure}
     \centering
     \begin{subfigure}[b]{0.235\textwidth}
         \centering
         \includegraphics[height=4.8cm]{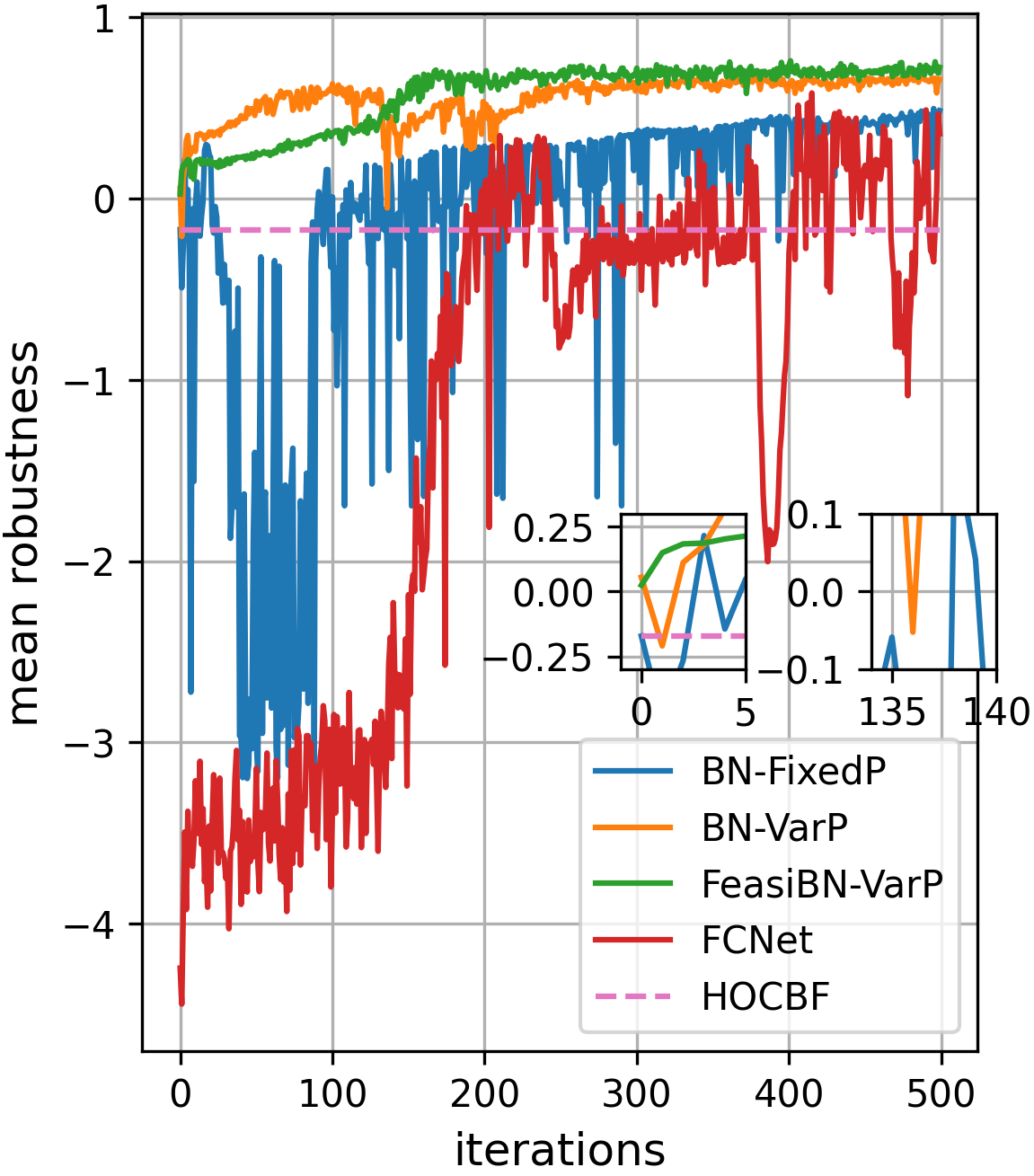}
         \caption{\small Robustness}
         \label{fig:ro}
     \end{subfigure}
     \ 
     \begin{subfigure}[b]{0.235\textwidth}
         \centering
         \includegraphics[height=4.8cm]{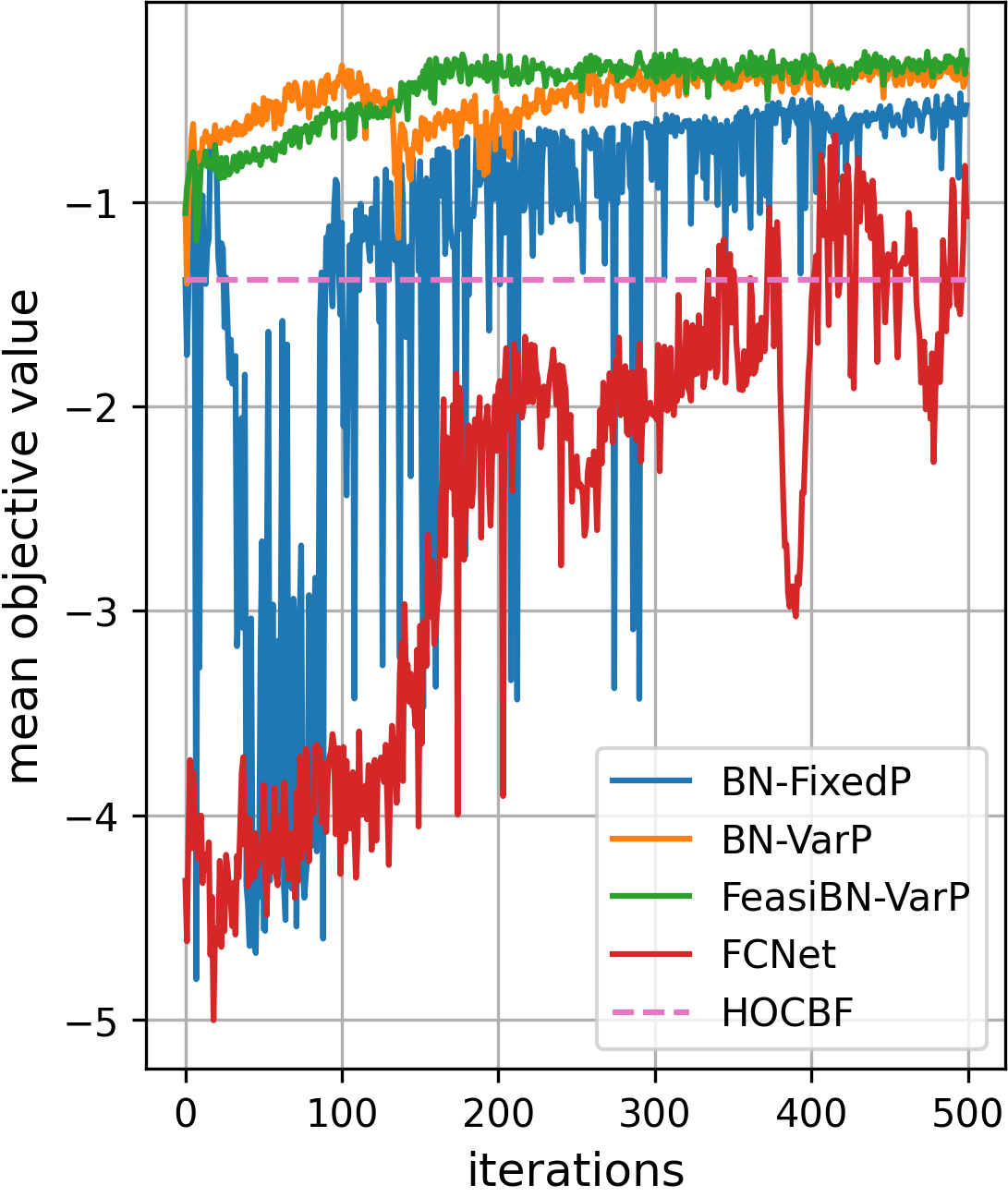}
         \caption{\small Objective}
         \label{fig:ob}
     \end{subfigure}
    \caption{Learning curves for BN-FixedP, BN-VarP, FeasiBN-VarP, and FCNet. Dashed lines show the results of directly using HOCBFs. (a) mean robustness $\rho(\varphi \wedge \varphi^{\mathcal U},\mathbf x, 0)$ during training. (b) mean objective values during training. }
\label{fig:curves}
\end{figure}
\begin{figure}
    \centering
    \begin{subfigure}[b]{0.24\textwidth}
        \centering
        \includegraphics[width=\textwidth]{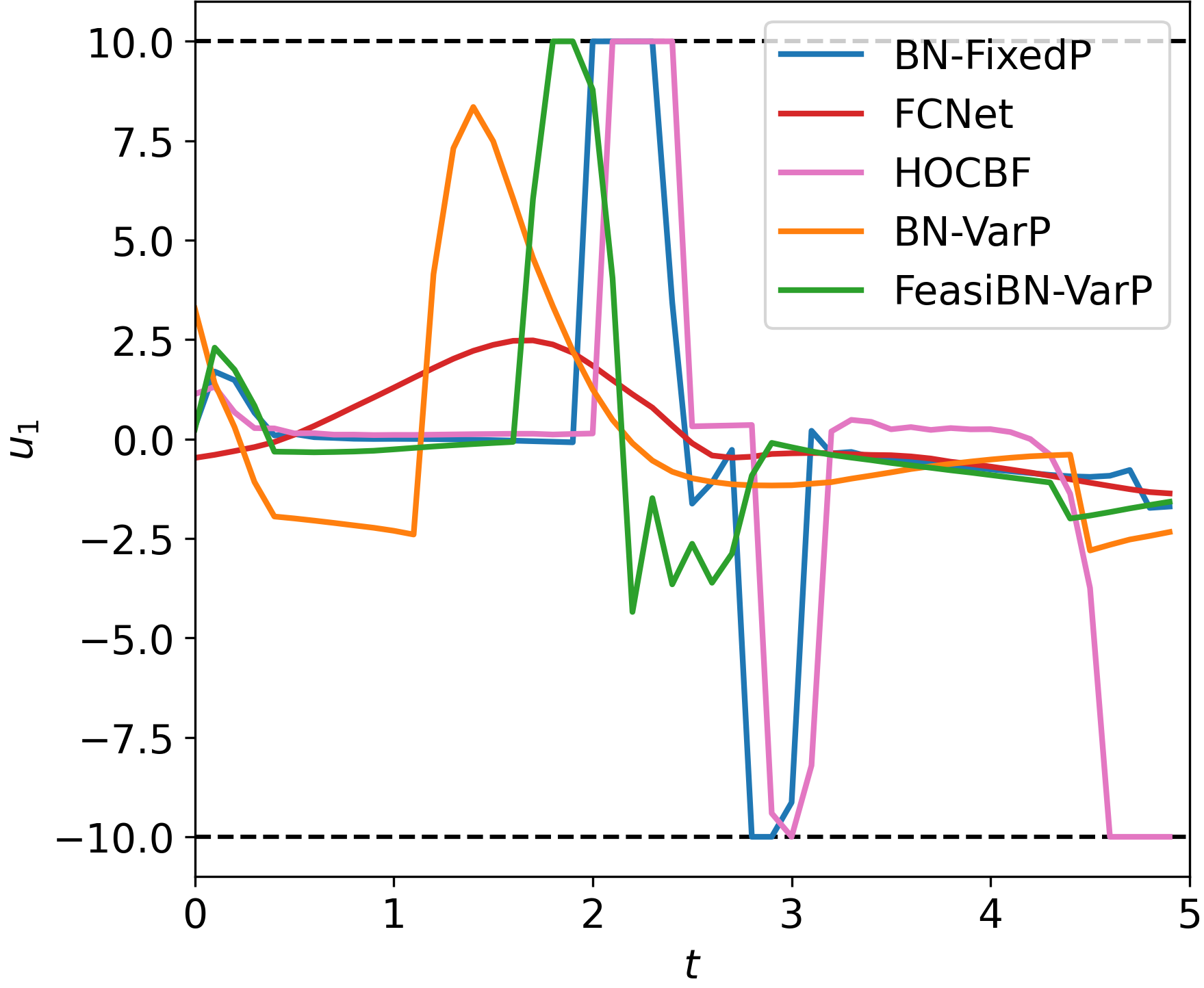}
        \caption{\small $u_1(t)$}
        \label{fig:input1}
    \end{subfigure}
    \hfill
    \begin{subfigure}[b]{0.24\textwidth}
        \centering
        \includegraphics[width=\textwidth]{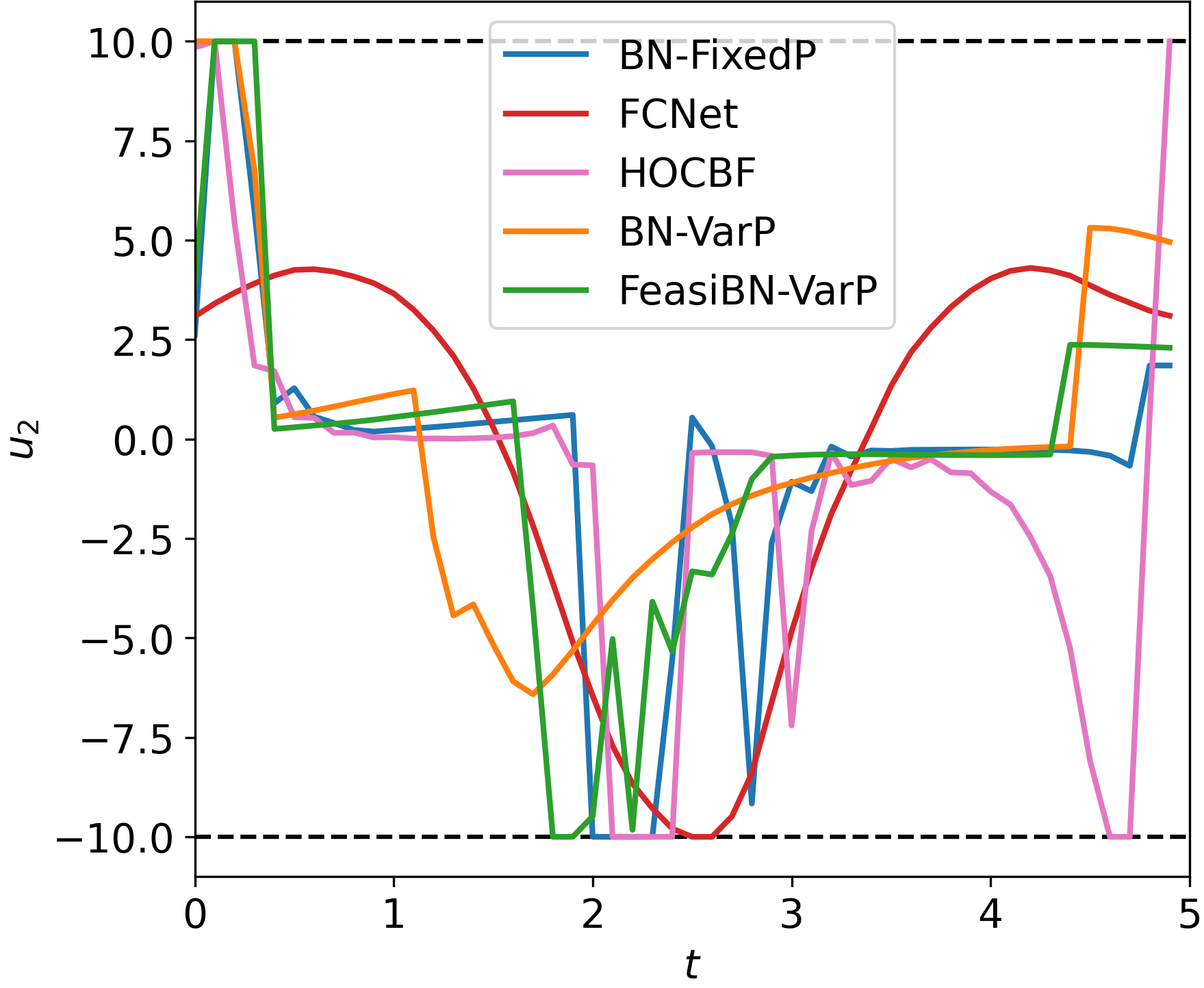}
        \caption{\small $u_2(t)$}
        \label{fig:input2}
    \end{subfigure}
    \caption{Inputs over time for BN-FixedP, BN-VarP, FeasiBN-VarP, FCNet and HOCBF.}
\label{fig:inputs}
\end{figure}
\noindent\textbf{Analysis and Discussion.} From Fig.~\ref{fig:ro}, we observe that when using BarrierNet (BN-FixedP, BN-VarP, and FeasiBN-VarP), all three methods achieve high robustness after sufficient training iterations. Specifically, BN-FixedP reaches a maximum robustness of approximately 0.45 during 500 iterations, while BN-VarP and FeasiBN-VarP reach about 0.6 and 0.7, respectively. Moreover, the robustness obtained by each BarrierNet-based method is higher than that of FCNet at nearly every iteration, demonstrating the correctness of Theorem \ref{thm:fea-stl}.
Occasionally, FCNet achieves higher robustness than BN-FixedP. This occurs because fixed multipliers limit the performance of BarrierNet, especially when the STL specification involves control bounds. In contrast, the use of time-varying multipliers in BN-VarP and FeasiBN-VarP significantly improves robustness, confirming that time-varying multipliers enhance the satisfaction of STL specifications, as claimed in Rem. \ref{rem:rem1}.
We further note that BN-VarP produces negative robustness at certain iterations, and further inspection reveals that the controller fails to satisfy the control-bound–related specification at those times. In comparison, FeasiBN-VarP maintains positive robustness throughout training, showing that the feasibility-aware formulation effectively guides the QP toward satisfying all constraints.
\begin{figure*}[t]
    \centering
    \begin{subfigure}[t]{0.45\linewidth}
        \centering        \includegraphics[width=1.0\linewidth]{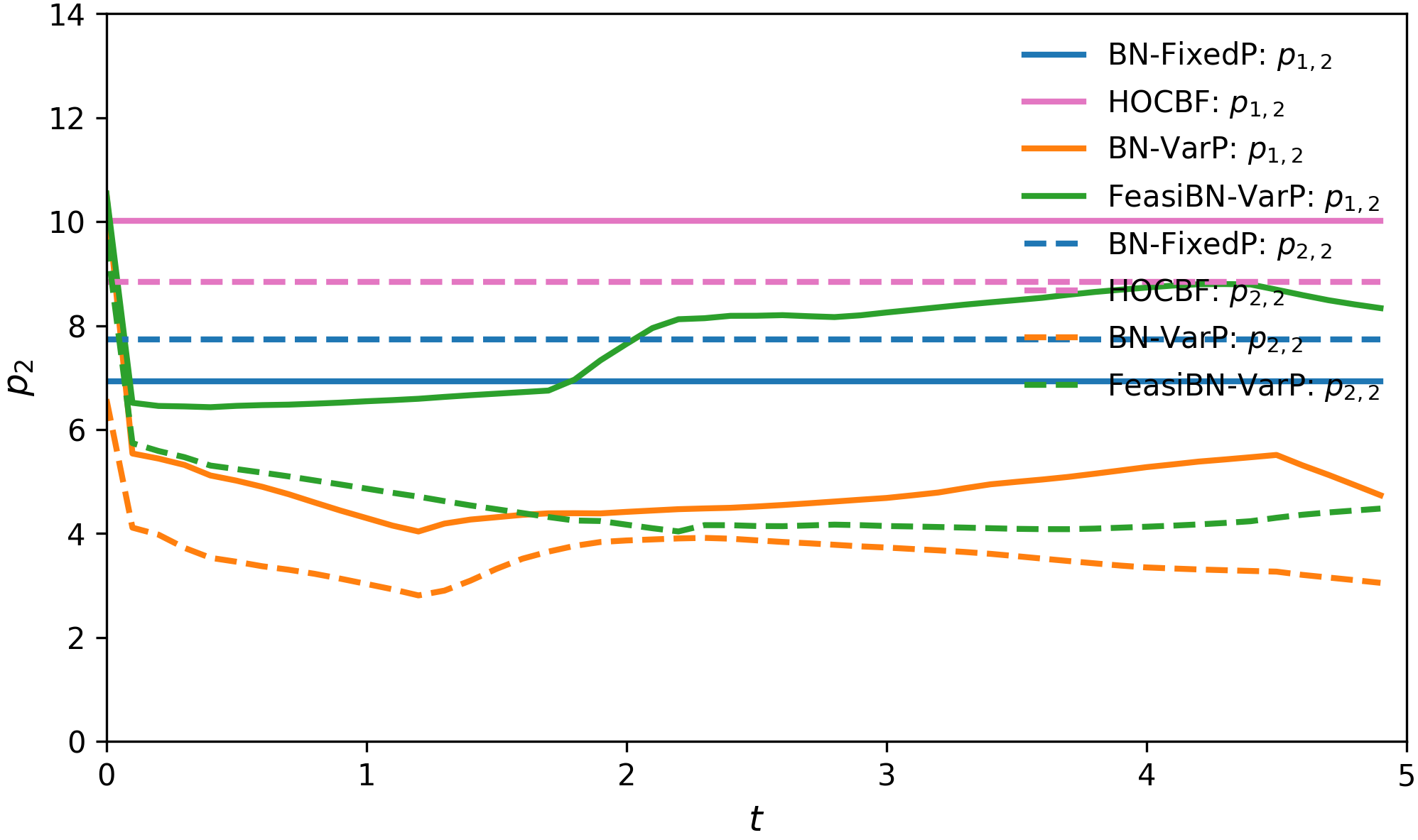}
        \caption{$p_{1,2}(t),\ p_{2,2}(t)$ for $F_{[2,5]}Reg_2$}
        \label{subfig:p2}
    \end{subfigure}
    \begin{subfigure}[t]{0.45\linewidth}
        \centering        \includegraphics[width=1.0\linewidth]{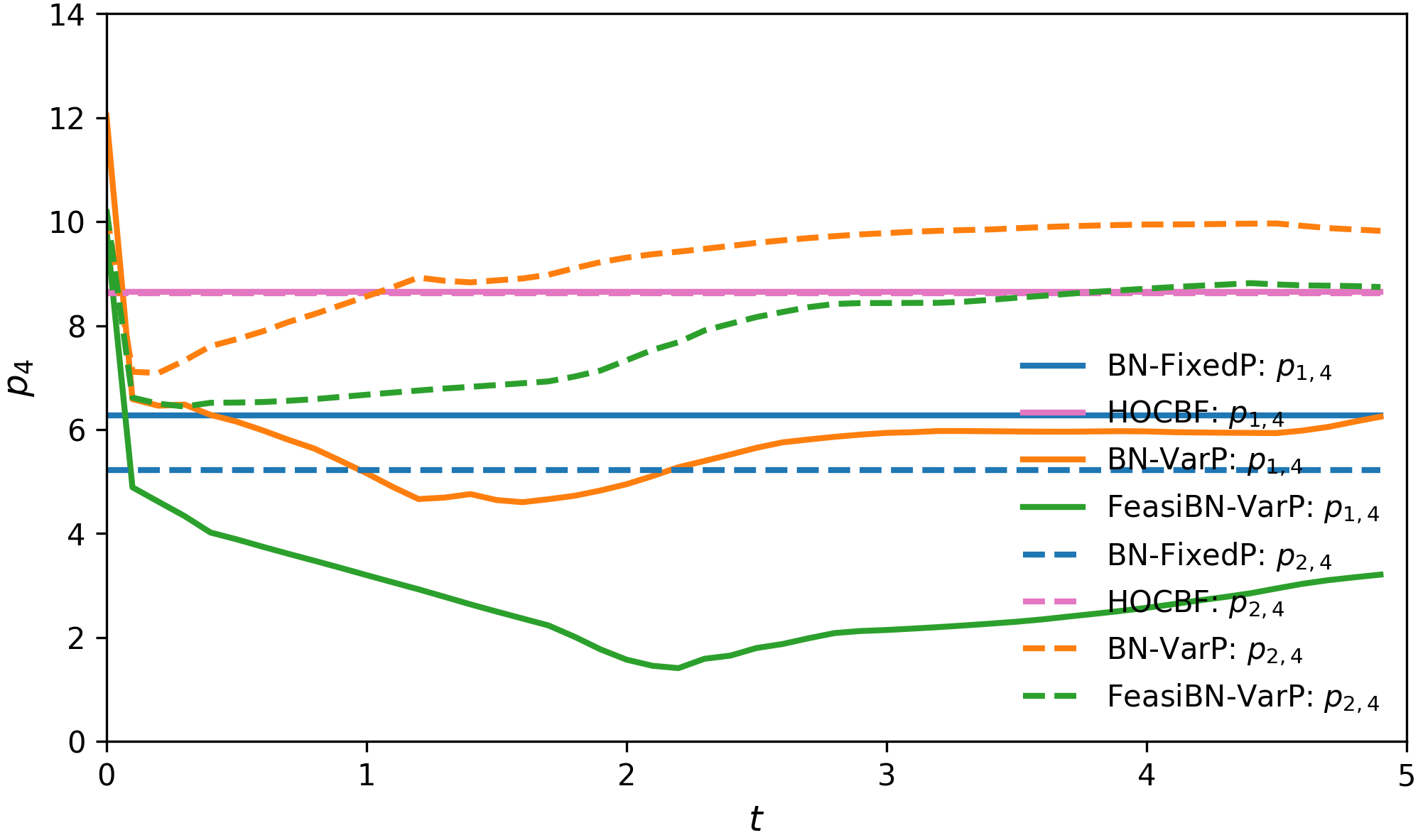}
        \caption{$p_{1,4}(t),\ p_{2,4}(t)$ for $G_{[0,5]}(\neg Obs_4)$}
        \label{subfig:p4}
    \end{subfigure}
    \caption{Multipliers over time for BN-FixedP, BN-VarP, FeasiBN-VarP, FCNet and HOCBF.}
    \label{fig:multipliers}
\end{figure*}

The results of directly applying HOCBFs with randomly chosen hyperparameters are similar to those obtained by using BN-FixedP with an untrained neural network (i.e., at the first training iteration). Due to the inter-sampling issue discussed in Rem. \ref{rem:rem2}, this approach does not satisfy the specification well: the robustness is negative, and although the robot reaches $Reg_1$, it fails to reach $Reg_2$. Consequently, this baseline is less robust than the trained BarrierNet and FCNet methods. As shown in Fig.~\ref{fig:traj}, after training, the robot reaches the centers of $Reg_1$ and $Reg_2$ in the correct order using both BarrierNet and FCNet. However, several trajectories produced by FCNet and BN-FixedP pass through $Obs_4$, whereas all trajectories generated by BN-VarP and FeasiBN-VarP satisfy every specification. This again demonstrates the superior robustness provided by time-varying multipliers in BarrierNet.

In both Fig.~\ref{fig:inputs} and Fig.~\ref{fig:multipliers}, we plot the input and multiplier profiles from a single representative trajectory chosen from the 10 trajectories. From Fig.~\ref{fig:inputs}, we observe that after training, both the BarrierNet-based methods and FCNet produce control inputs that satisfy the input bounds. The input profiles of BN-FixedP closely resemble those of the HOCBF baseline after training, as Fig.~\ref{fig:multipliers} shows that the learned multipliers of BN-FixedP are constant over time and close to the fixed HOCBF multipliers after training. These results indicate that fixed multipliers substantially limit the ability of BarrierNet to improve performance. In contrast, BN-VarP and FeasiBN-VarP adjust the inputs earlier and more aggressively to satisfy the specifications, as shown in Fig.~\ref{fig:input2}. This behavior is consistent with Fig.~\ref{fig:multipliers}, where the time-varying multipliers adapt across time steps according to the control needs. Such adaptability enables these methods to achieve higher robustness.

\subsubsection{Neural Network Controller with Memory}
\label{subsubsec:caseI2}
Consider the environment shown in Fig. \ref{fig:trajectories2}. $\mathbf x_0$ is uniformly sampled in the region $Init$ with zero velocity. We discretize the system with a time interval of $0.1s$. The task for the robot is given by an STL formula:
\begin{equation}
    \label{eq:task2}
    \varphi = F_{[0,5]}Reg_1 \land F_{[5,10]}Reg_2 \land G_{[0,10]}(\neg Obs_3 \land \neg Obs_4 \land\neg Obs_5),
\end{equation}
where $Reg_j$ indicates $R_j-\|l(\mathbf x) - \mathbf o_j\|_2\geq 0$, $j=1,2$, $l(\mathbf x) = [x\ y]^\top$. The control bounds given by an STL formula is
\begin{equation}
\label{eq:pho-u3}
\varphi^{\mathcal U}
= 
G_{[0,10]}\big((\mathbf{u}-\mathbf{u}_{\min}\ge \mathbf{0})
\wedge
(\mathbf{u}_{\max}-\mathbf{u}\ge \mathbf{0})\big),
\end{equation}
where $\mathbf{u}_{\min}=[-10\ -10]^{\top},$ and $\mathbf{u}_{\max}=[10\ 10]^{\top}.$ $Obs_j$ is a superellipse defined by \eqref{eq:obs} where $j=3,4,5$.
Here, $Reg_j$ belongs to Category II and $Obs_j$ belongs to Category I for $j=1,2,3,4,5$. In plain terms, the STL formula $\varphi$ requires the robot to reach $Reg_1$ within the time interval $[0,5]$ and $Reg_2$ within $[5,10]$, while avoiding obstacles $Obs_3$, $Obs_4$ and $Obs_5$ at all times. The time horizon of $\varphi$ is $10$. All five predicates have a relative degree of $2$ with respect to system~\eqref{eq:robot}. In this example, we set $\mathbf Q(\cdot, \bm\theta_q)$ to be an identity matrix, so for all $t>0$ the output of the previous layers reduces to $\mathbf F(\cdot, \bm\theta_f)$, which serves as a reference control.
Because the task requires the robot to visit $Reg_1$ and then return to $Reg_2$ (within the $Init$ region), the motion involves a back-and-forth behavior. Therefore, we evaluate both memory-based and memoryless structures for generating $\mathbf F$.
For the memoryless structure, RefNet is implemented as a feedforward neural network with three fully connected layers.
For the memory-based structure, RefNet is implemented as a 2-layer recurrent neural network (LSTM).
In both cases, InitNet and BarrierNet are implemented as feedforward neural networks with three fully connected layers each.
For the robustness measure, we adopt the exponential robustness in \eqref{eq:pho2} and \eqref{eq:pho3}.

\noindent\textbf{Comparison setup.} We construct the HOCBFs \eqref{eq:cat1}, \eqref{eq:cat23} and train the controller proposed in this paper under two different configurations. We directly adopt the FeasiBN-VarP configuration described in Sec. \ref{subsubsec:caseI1}. 
For comparison, we evaluate two types of reference inputs generated by RefNet: 
a memoryless FCNet producing $\mathbf{F}(\mathbf{x}(t), \boldsymbol{\theta}_f)$, 
and a memory-based RNN producing $\mathbf{F}(\mathbf{x}_{0:t}, \boldsymbol{\theta}_f)$. 
To ensure a fair comparison, we use the same objective function, optimizer, 
and InitNet architectures across all experiments.

\begin{figure}[t]
    \centering
    \begin{subfigure}[b]{0.24\textwidth}
        \centering        \includegraphics[width=\textwidth]{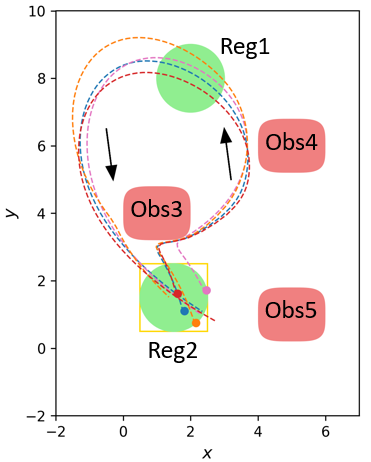}
        \caption{\small FCNet-generated  $\mathbf{F}(\mathbf{x}(t), \boldsymbol{\theta}_f)$}
        \label{fig:traj2}
    \end{subfigure}
    \hfill
    \begin{subfigure}[b]{0.24\textwidth}
        \centering        \includegraphics[width=\textwidth]{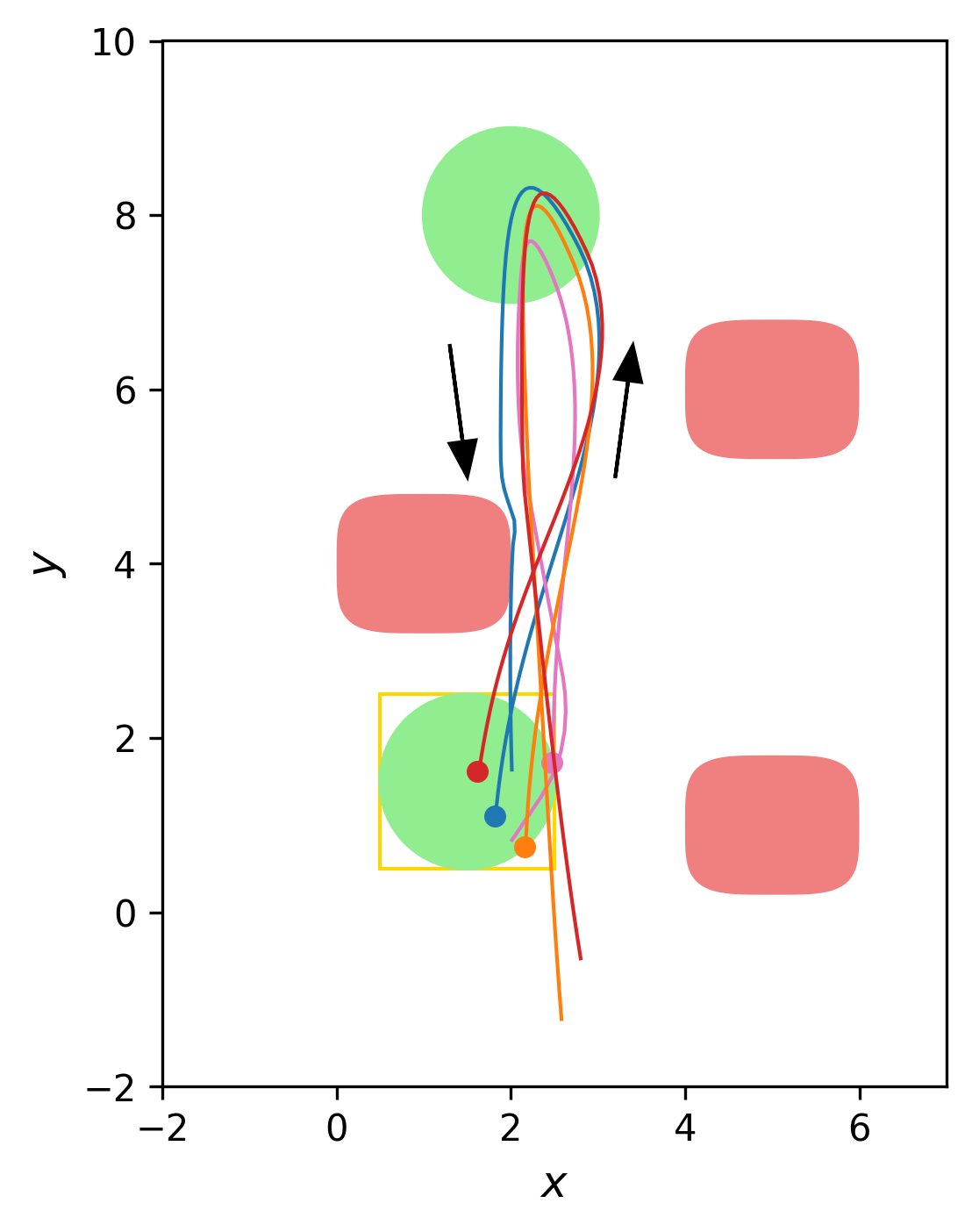}
        \caption{\small RNN-generated  $\mathbf{F}(\mathbf{x}_{0:t}, \boldsymbol{\theta}_f)$}
        \label{fig:traj3}
    \end{subfigure}
    \caption{4 trajectories from
sampled initial conditions under FeasiBN-VarP (a) without memory and (b) with memory.}
\label{fig:trajectories2}
\end{figure}
\noindent\textbf{Analysis and Discussion.} From Fig.~\ref{fig:trajectories2}, we observe that both the memoryless and the memory-based neural controllers satisfy the STL specifications. For the memoryless controller (Fig.~\ref{fig:traj2}), the robot leaves the initial area, avoids $Obs_3$ and $Obs_4$, reaches $Reg_1$, and then approaches $Reg_2$ from the opposite side of $Obs_3$, resulting in a relatively long overall trajectory. In contrast, under the memory-based controller (Fig.~\ref{fig:traj3}), the robot reaches $Reg_1$ in a similar manner but subsequently moves to $Reg_2$ from the same side of $Obs_3$, producing a noticeably shorter trajectory.

This difference arises because, for back-and-forth motions, a controller equipped with state-history memory can implicitly retain information about the previously visited regions and the path taken to reach them. By leveraging this temporal context, the memory-based network can reason about where the robot has already been and avoid unnecessary detours when planning the return path. In contrast, a purely memoryless controller relies only on the instantaneous state and lacks awareness of past motion, making it more prone to selecting longer but still feasible routes. The input profiles in Fig. \ref{fig:inputs2} reveal a clear difference between the memoryless and memory-based controllers. The dashed curves (memoryless) exhibit sharp oscillations, large instantaneous jumps, and frequent excursions toward the input limits. In contrast, the solid curves (memory-based) demonstrate smoother and more regulated input evolutions. This indicates that incorporating state-history memory enables the controller to make more informed adjustments, avoiding abrupt reactions to local state changes. As a result, the memory-based controller tends to reduce unnecessary control effort while maintaining stable and feasible inputs throughout the trajectory.

The memory capabilities are particularly beneficial for STL specifications that involve sequential or repetitive tasks, such as visiting regions in a prescribed order (e.g., visiting $Reg_1$ first and then $Reg_2$), alternating reachability requirements (back-and-forth missions), tasks that require remembering whether a region has already been visited, and avoiding previously traversed areas (e.g., “avoid returning to the same unsafe side”). For these specifications, the temporal dependencies cannot be inferred from a single snapshot of the current state, and therefore memory-based controllers can produce more efficient trajectories while still satisfying STL constraints.

\begin{figure}[t]
    \centering
    \begin{subfigure}[b]{0.2\textwidth}
        \centering
        \includegraphics[width=\textwidth]{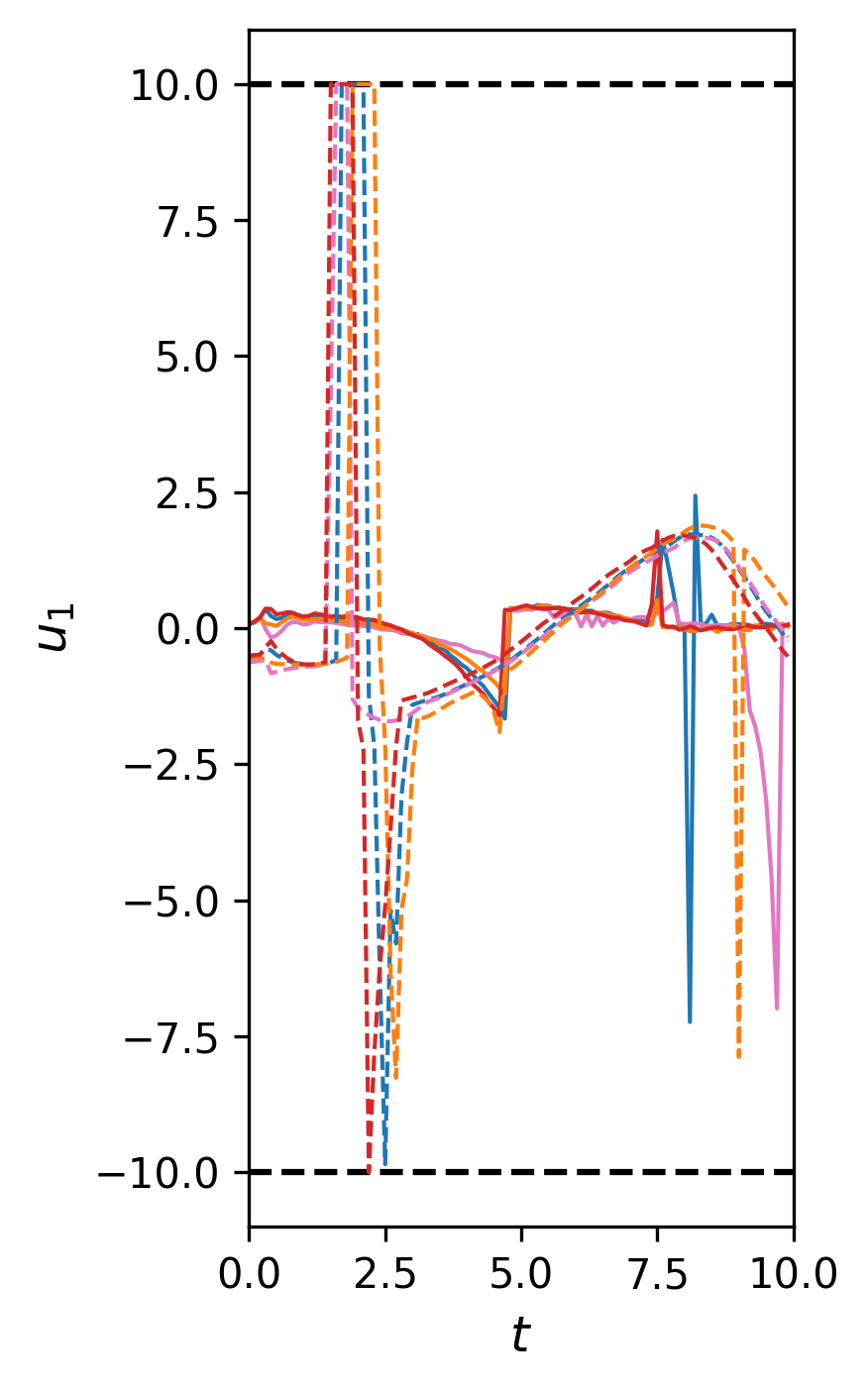}
        \caption{\small $u_1(t)$}
        \label{fig:input3}
    \end{subfigure}
    \hspace{0.02\textwidth}
    \begin{subfigure}[b]{0.2\textwidth}
        \centering
        \includegraphics[width=\textwidth]{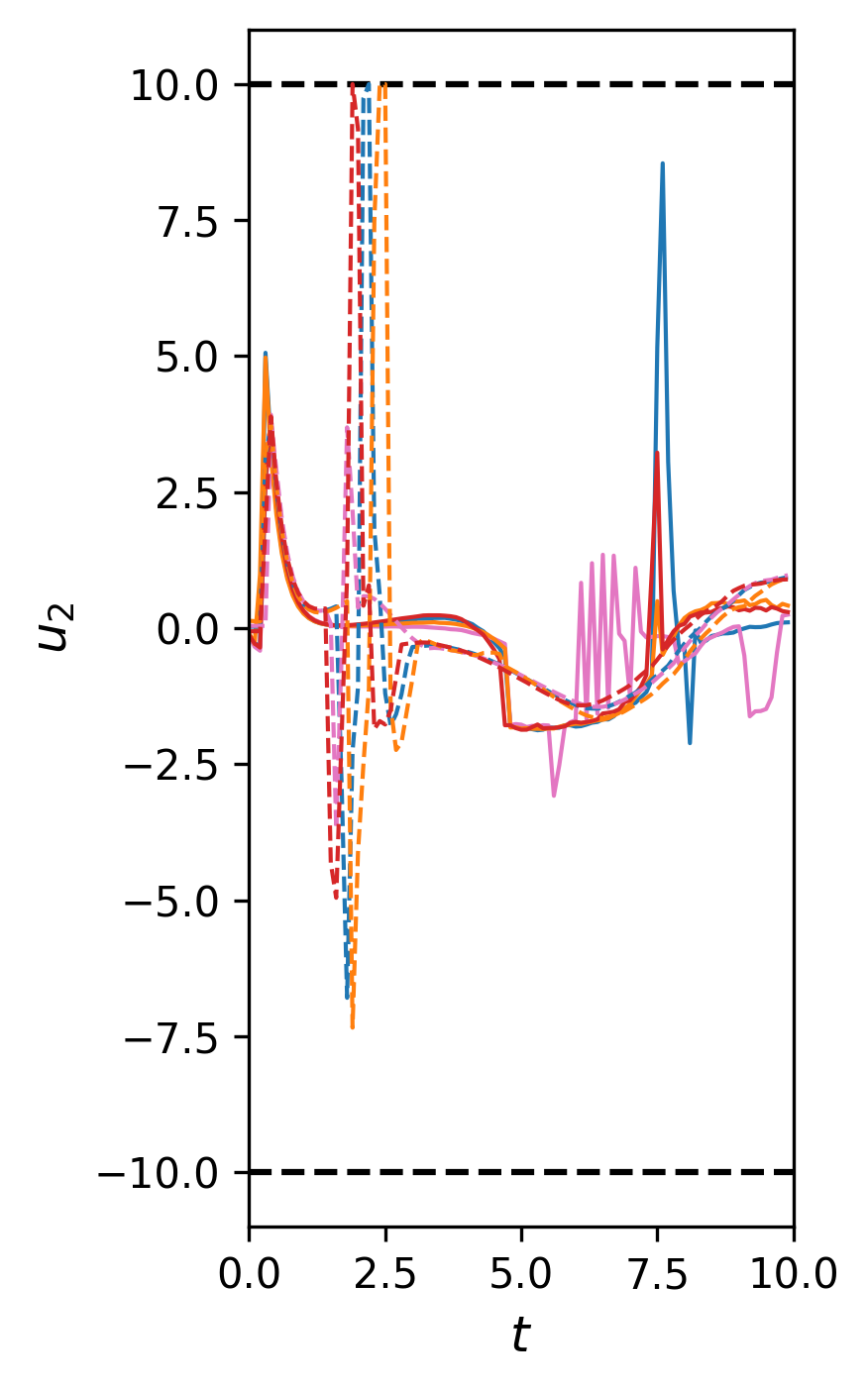}
        \caption{\small $u_2(t)$}
        \label{fig:input4}
    \end{subfigure}
    \caption{Inputs over time for 4 trajectories from sampled initial conditions under FeasiBN-VarP (solid curves: controller with memory; dashed curves: controller without memory).}
    \label{fig:inputs2}
\end{figure}

\subsection{Case Study II: Nonlinear Dynamical System — Unicycle} 
Consider a robot with unicycle model in the form:
\begin{equation}
    \label{eq:robot2}
    \begin{bmatrix} \dot{x}\\ \dot{y}\\ \dot{\theta}\\ \dot{v}\end{bmatrix} = \begin{bmatrix} v\cos(\theta) \\ v\sin(\theta)\\ 0\\ 0\end{bmatrix} + \begin{bmatrix} 0&0\\ 0&0\\ 1&0\\ 0&1\end{bmatrix} \begin{bmatrix}u_1\\u_2\end{bmatrix},
\end{equation}
where $\mathbf{x} = [x\ y\ \theta\ v]^\top$ and $\mathbf{u} = [u_1\ u_2]^\top$, with $[x\ y]^\top$ denoting the 2D position, $\theta$ the heading angle, $v$ the linear speed, $u_1$ the angular velocity, and $u_2^\top$ the acceleration of the robot. The control input treatment and cost function settings follow those in Sec.~\ref{subsec:case 1}.

Consider the environment shown in Fig. \ref{fig:uni-fc}. $\mathbf x_0$ is uniformly sampled in the region $Init$ with zero velocity. We discretize the system with a time interval of $0.1s$. The task for the robot is given by an STL formula:
\begin{equation}
    \label{eq:task3}
    \varphi = F_{[0,3]}Reg_1 \land F_{[3,6]}Reg_2 \land G_{[0,6]}(\neg Obs_3 \land \neg Obs_4 \land \neg Obs_5),
\end{equation}
where $Reg_j$ indicates $R_j-\|l(\mathbf x) - \mathbf o_j\|_2\geq 0$, $j=1,2$, $l(\mathbf x) = [x\ y]^\top$. The control bounds given by an STL formula is
\begin{equation}
\label{eq:pho-u4}
\varphi^{\mathcal U}
= 
G_{[0,6]}\big((\mathbf{u}-\mathbf{u}_{\min}\ge \mathbf{0})
\wedge
(\mathbf{u}_{\max}-\mathbf{u}\ge \mathbf{0})\big),
\end{equation}
where $\mathbf{u}_{\min}=[-10\ -10]^{\top},$ and $\mathbf{u}_{\max}=[10\ 10]^{\top}$. $Obs_j$ is a superellipse defined by \eqref{eq:obs} where $j=3,4,5$.
Here, $Reg_j$ belongs to Category II and $Obs_j$ belongs to Category I for $j=1,2,3,4,5$.
In plain terms, the STL formula $\varphi$ requires the robot to reach $Reg_1$ within the time interval $[0,3]$ and $Reg_2$ within $[3,6]$, while avoiding obstacles $Obs_3$, $Obs_4$ and $Obs_5$ at all times.
The time horizon of $\varphi$ is $6$.
All five predicates have a relative degree of $2$ with respect to system~\eqref{eq:robot}. In this example, we set $\mathbf Q(\mathbf x(t), \bm\theta_q)$ to be an identity matrix, so for all $t>0$ the output of the previous layers reduces to $\mathbf F(\mathbf x(t), \bm\theta_f)$, which serves as a memoryless reference control. We directly adopt the FeasiBN-VarP configuration described in Sec.~\ref{subsubsec:caseI1}. RefNet and InitNet are implemented as feedforward neural networks with three fully connected layers each.
For the robustness measure, we adopt the exponential robustness in \eqref{eq:pho2} and \eqref{eq:pho3}.

\noindent\textbf{Analysis and Discussion.} In Fig.~\ref{fig:uni-fc}, the 10 trajectories illustrate how FeasiBN-VarP behaves under different sampled initial conditions. All trajectories avoid the obstacles and reach the target regions, demonstrating reliable satisfaction of the STL specification. The method adaptively utilizes the available free space by choosing obstacle-avoiding routes that best match each initial configuration. For example, the trajectories mainly differ in how they pass the region between $Reg_1$ and $Reg_2$: all of them move above $Obs_3$, but some travel closer to the upper boundary of $Obs_3$ while others take slightly higher arcs before turning toward $Reg_2$. These variations arise from different initial positions, leading FeasiBN-VarP to exploit nearby free space differently while still producing valid specification-satisfying paths. Overall, the results illustrate robust and consistent performance across diverse initial conditions.

\begin{figure}
    \centering
\includegraphics[width=7cm]{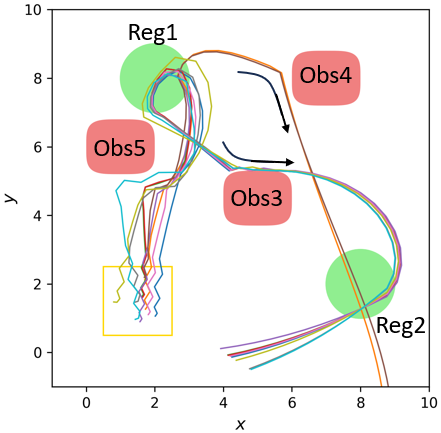}
    \caption{10 trajectories with sampled initial conditions using FeasiBN-VarP. FeasiBN-VarP can flexibly generate trajectories that satisfy the specification by leveraging the available space under different initial conditions.}
    \label{fig:uni-fc}
\end{figure}

The learning curves in Fig.~\ref{subfig:rob-uni} show that the robustness remains strictly positive throughout training, confirming that FeasiBN-VarP consistently produces specification-satisfying trajectories. Compared with Fig.~\ref{fig:ro}, however, the robustness evolution is noticeably more abrupt, with larger fluctuations across iterations. This is likely due to the nonlinear unicycle dynamics and the more complex obstacle-rich environment, which make the optimization landscape less smooth and amplify sensitivity to small parameter updates. The control inputs $u_1(t)$ and $u_2(t)$ remain strictly within their bounds for all trajectories in Figs.~\ref{subfig:input1-uni} and \ref{subfig:input2-uni}, demonstrating that the input constraints are always respected. Nevertheless, both inputs exhibit pronounced oscillations during the initial phase of each trajectory. These oscillations correspond to a warm-up period in which the controller adjusts the heading direction, as the robot must quickly orient itself from diverse initial states before executing the main task. This transient alignment process explains the zigzag behavior observed in $x(t)$ and $\theta(t)$ at the beginning of the trajectories in Figs.~\ref{subfig:x-uni} and \ref{subfig:theta-uni}. The trajectories similarly show mild early-stage zigzag motion as the system resolves these heading inconsistencies. Such transient behaviors are not problematic; they simply reflect the inherent need to reorient the unicycle before moving toward the target regions. After this brief adjustment phase, both the heading and the state trajectories become much more coherent and aligned across runs, indicating stable closed-loop performance.
\begin{figure*}[t]
    \centering
    \begin{subfigure}[t]{0.19\linewidth}
        \centering
        \includegraphics[width=1.0\linewidth]{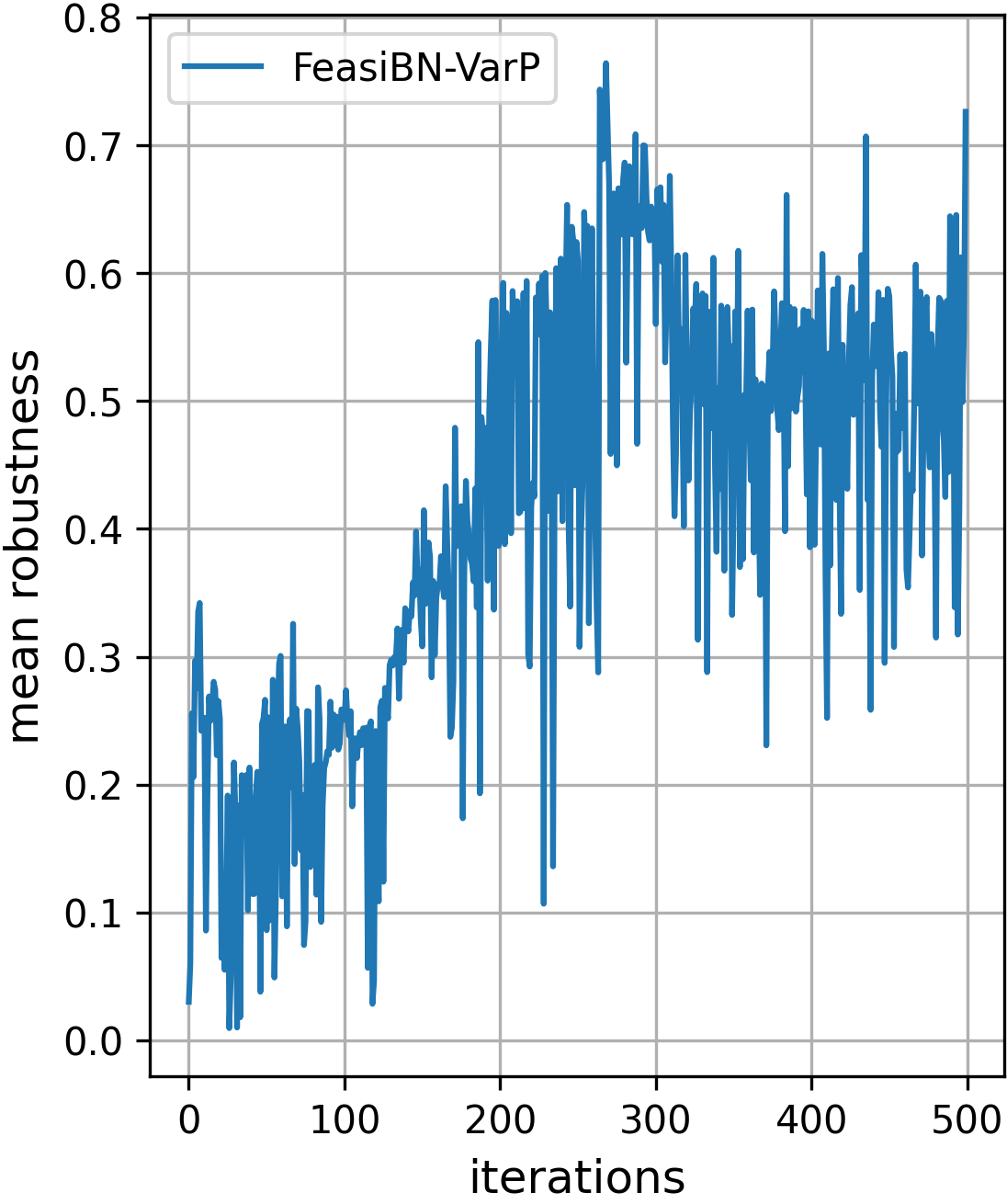}
        \caption{Robustness}
        \label{subfig:rob-uni}
    \end{subfigure}
    \begin{subfigure}[t]{0.195\linewidth}
        \centering
        \includegraphics[width=1.0\linewidth]{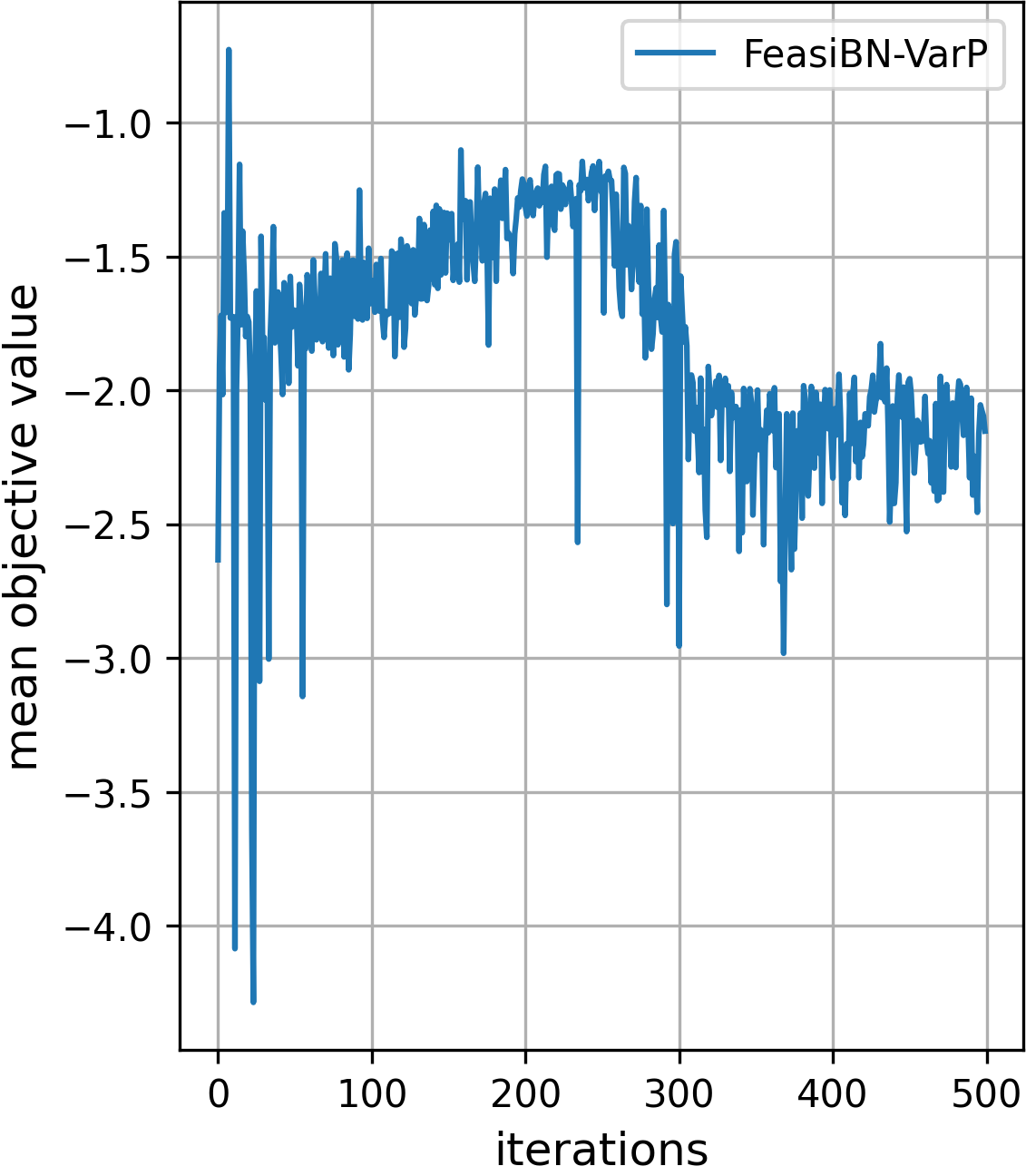}
        \caption{Objective}
        \label{subfig:obj-uni}
    \end{subfigure}  
    \begin{subfigure}[t]{0.29\linewidth}
        \centering
        \includegraphics[width=1.0\linewidth]{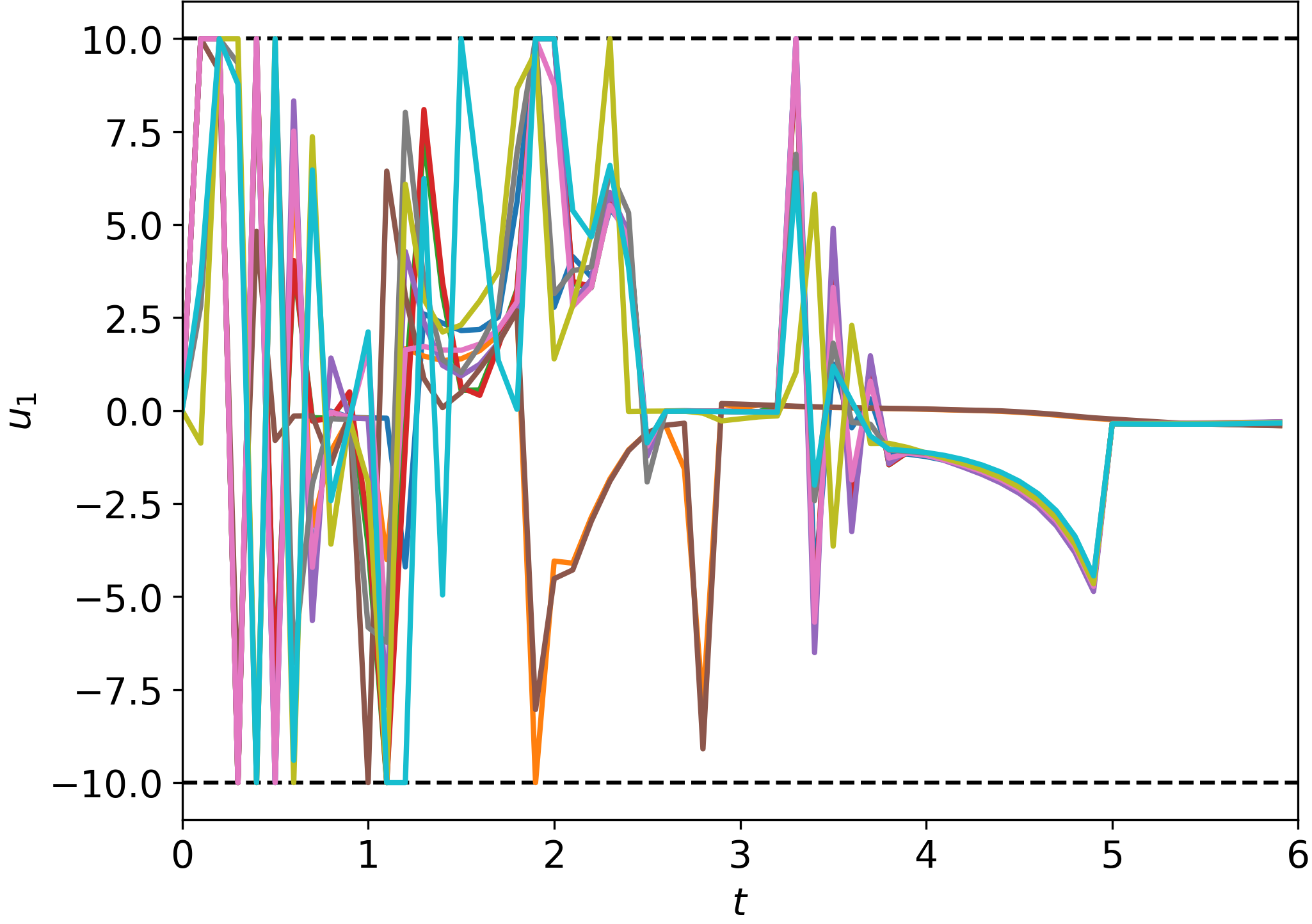}
        \caption{$u_{1}(t)$}
        \label{subfig:input1-uni}
    \end{subfigure}
    \begin{subfigure}[t]{0.29\linewidth}
        \centering
        \includegraphics[width=1.0\linewidth]{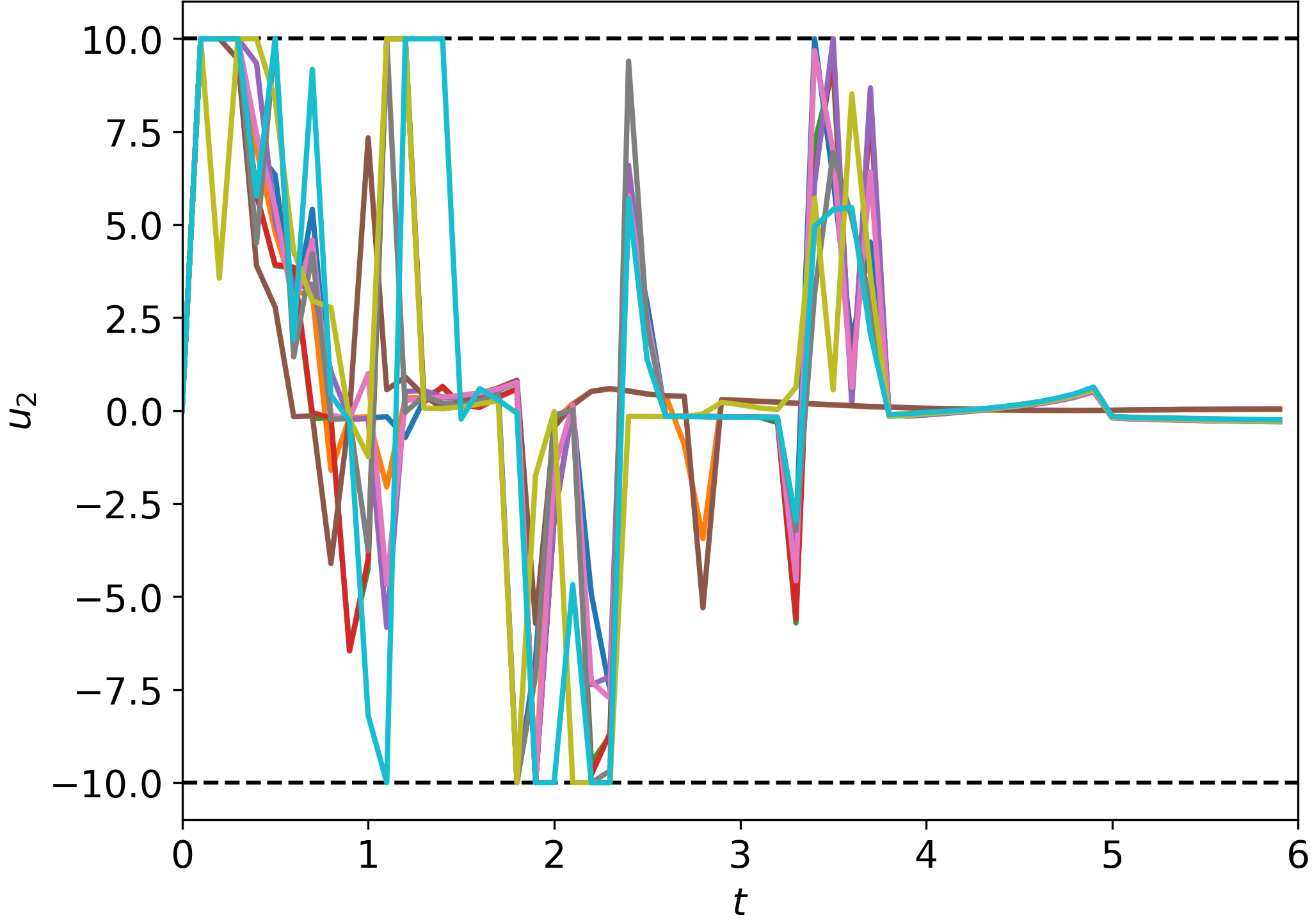}
        \caption{$u_{2}(t)$}
        \label{subfig:input2-uni}
    \end{subfigure}
    \caption{Learning curves and inputs over time for 10 trajectories generated by FeasiBN-VarP. The robustness remains strictly positive and the inputs stay within their bounds, indicating that FeasiBN-VarP robustly satisfies all specifications.}
    \label{fig:uni1}
\end{figure*}

The time-varying multipliers in Fig.~\ref{fig:p-uni} illustrate how BN-FixedP adapts the constraint strength throughout the trajectory. The plotted multiplier profiles correspond to a single representative trajectory selected from the 10 trajectories. By allowing the multipliers to evolve over time, the controller can selectively tighten or relax individual HOCBF constraints when needed, which improves feasibility under tight input limits and enhances satisfaction of the STL specification. A notable feature is that $p_{1,1}$ and $p_{2,1}$ terminate before $t=3$, which results from the deletion rule applied once the first reachability task (entering $Reg_1$) is completed. After the robot reaches $Reg_1$, the corresponding predicate $b_1(x,t)$ is removed, and its associated multipliers are no longer active. This mechanism prevents unnecessary constraint enforcement and allows the controller to focus on the remaining STL requirements.
\begin{figure*}[t]
    \centering
    \begin{subfigure}[t]{0.24\linewidth}
        \centering
        \includegraphics[width=1.0\linewidth]{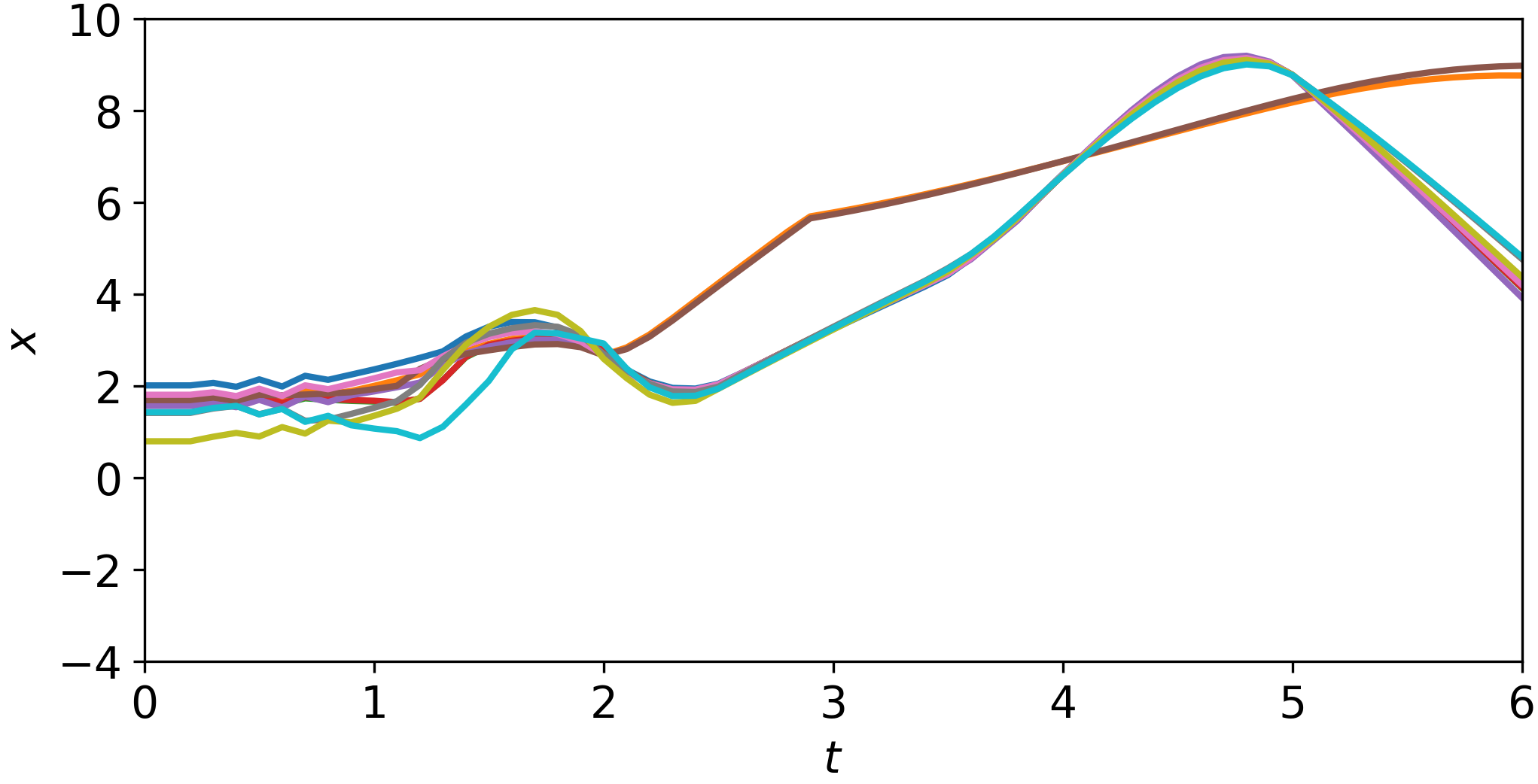}
        \caption{$x(t)$}
        \label{subfig:x-uni}
    \end{subfigure}
    \begin{subfigure}[t]{0.24\linewidth}
        \centering
        \includegraphics[width=1.0\linewidth]{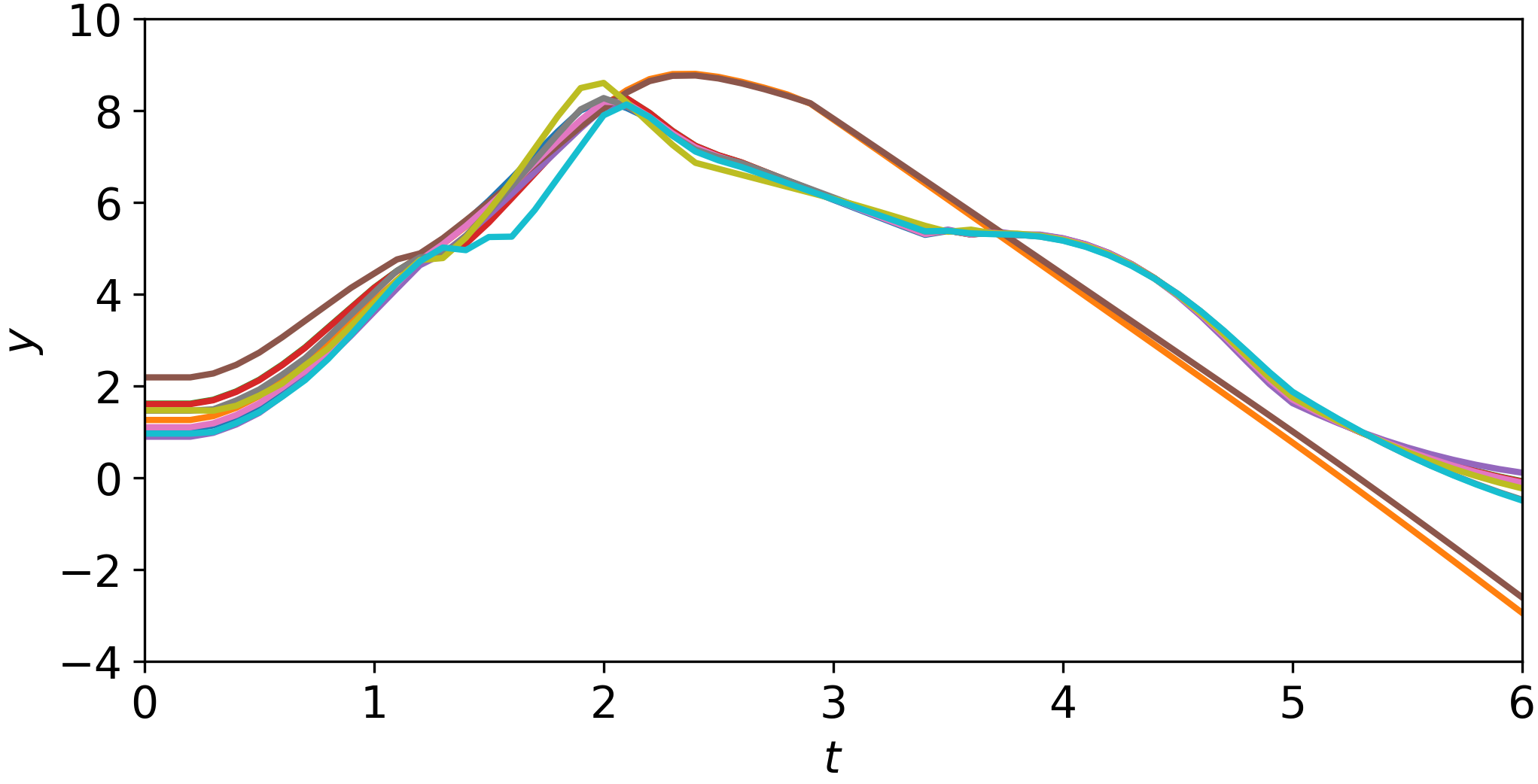}
        \caption{$y(t)$}
        \label{subfig:y-uni}
    \end{subfigure}  
    \begin{subfigure}[t]{0.24\linewidth}
        \centering
        \includegraphics[width=1.0\linewidth]{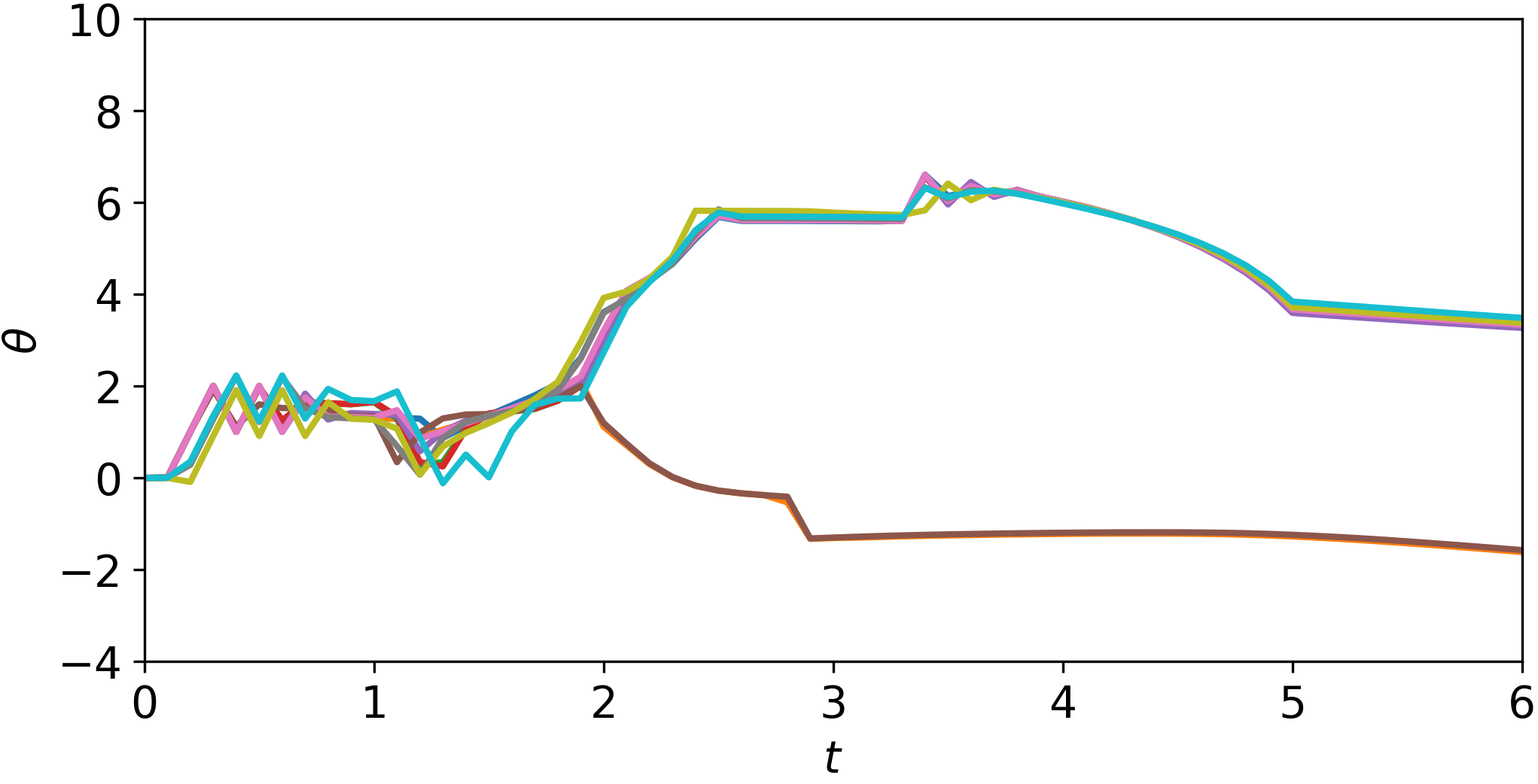}
        \caption{$\theta(t)$}
        \label{subfig:theta-uni}
    \end{subfigure}
    \begin{subfigure}[t]{0.24\linewidth}
        \centering
        \includegraphics[width=1.0\linewidth]{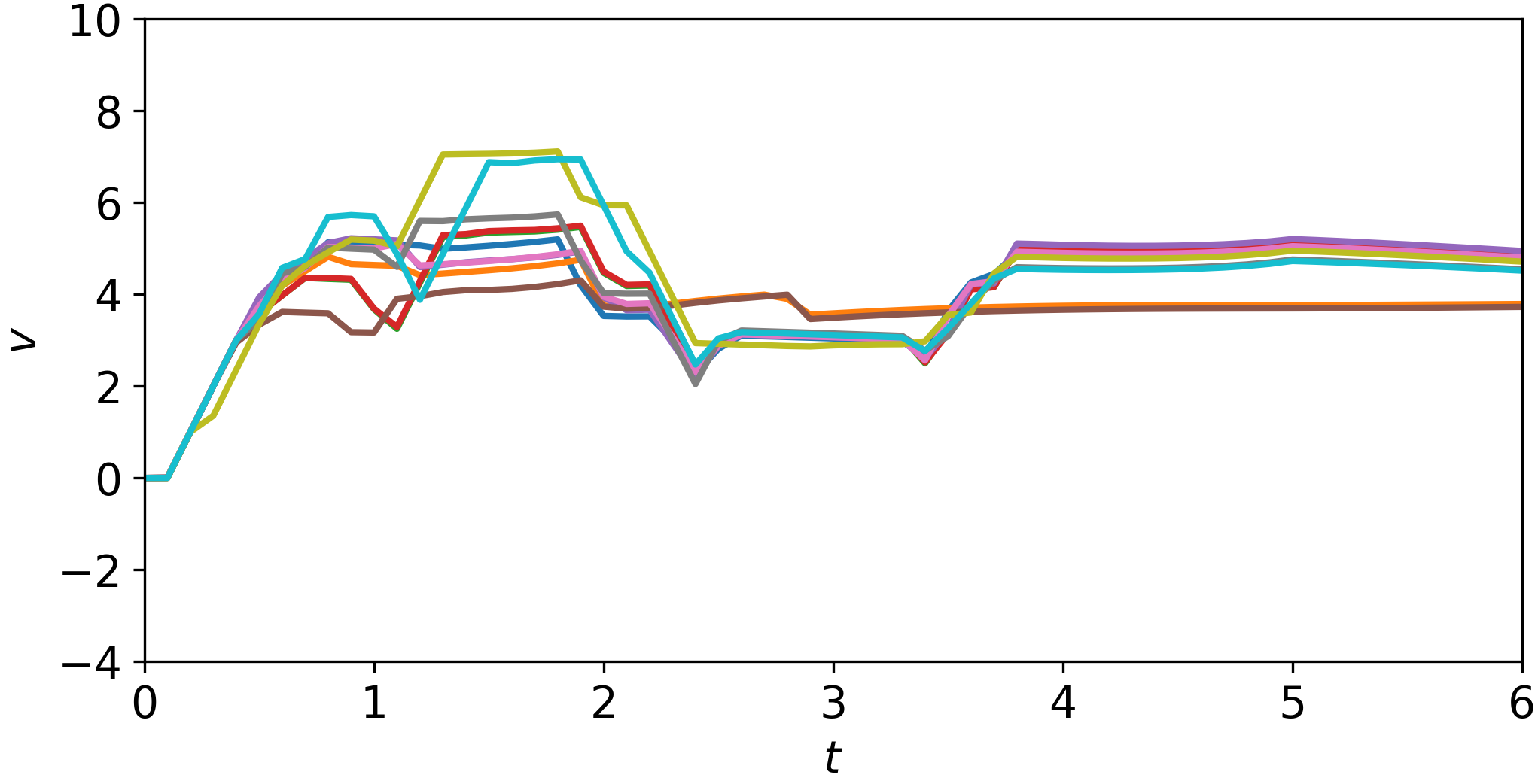}
        \caption{$v(t)$}
        \label{subfig:v-uni}
    \end{subfigure}
    \caption{States over time for 10 trajectories generated by FeasiBN-VarP. The initial zigzag behavior in $\theta(t)$ reflects transient heading adjustments, after which the trajectories become more aligned across runs. }
    \label{fig:uni3}
\end{figure*}

\begin{figure}
    \centering
\includegraphics[width=8cm]{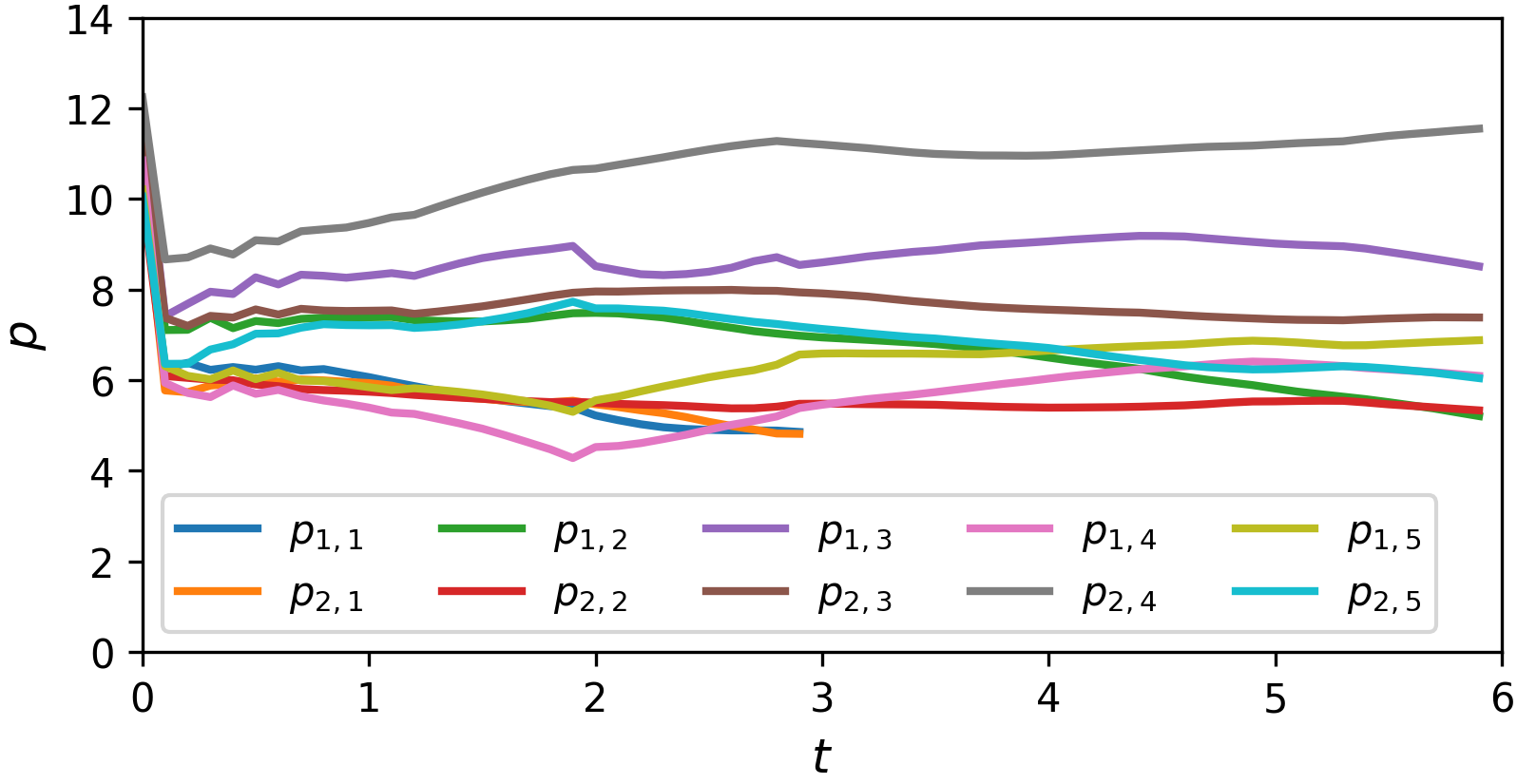}
    \caption{Multipliers over time for BN-FixedP. $p_{1,1}$ and $p_{2,1}$ terminate before $t=3$ due to the deletion rules described in Sec.~\ref{subsec:hocbf-stl}. 
Once the robot enters $Reg_1$, the corresponding $b_{1}(\mathbf{x},t)$ is removed.}
    \label{fig:p-uni}
\end{figure}

\section{Conclusion and Future Work}
\label{sec:conclusion}
In this paper, we proposed a feasibility-aware STL-BarrierNet framework for learning controllers that satisfy STL specifications while preserving the feasibility of the underlying optimization problem. By augmenting BarrierNet with time-varying HOCBF multipliers and a unified robustness metric that accounts for STL satisfaction, QP feasibility, and control bounds, the proposed method provides a systematic way to automatically tune constraint-related hyperparameters over time and across initial conditions. The resulting time-varying parametrization reduces conservativeness, avoids infeasible QPs under tight input limits, and improves closed-loop robustness in nonlinear and cluttered environments.

Several directions remain for future work. First, because the controller operates in discrete time, inter-sampling effects may still lead to constraint violations between updates, calling for approaches that explicitly account for sampled-data behavior. Second, extending the framework to multi-agent settings introduces challenges such as decentralized STL satisfaction, coordination among agents, and avoidance of inter-agent conflicts. Finally, real-world deployment requires handling disturbances and model mismatch, making robustness to uncertainty an important open problem.
\bibliographystyle{IEEEtran}
\bibliography{references.bib}

\vspace{10pt}

\setlength{\intextsep}{2pt}
\begin{wrapfigure}{l}{25mm} 
\includegraphics[width=1in,height=1.25in,clip,keepaspectratio]{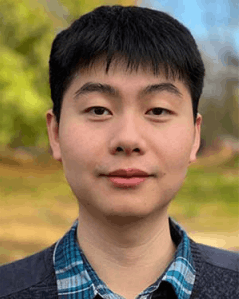}
\end{wrapfigure}\par
{\small \textbf{Wenliang Liu} (Member, IEEE) received the
B.Sc. degree in instrumentation science from
Beihang University, Beijing, China, in 2019, and
the M.Sc. and Ph.D. degrees in mechanical engineering from Boston University, Boston, MA,
USA, in 2023 and 2024.
His research interests include control theory,
formal methods, and machine learning, with
robotics applications. \par}

\vspace{10pt}

\setlength{\intextsep}{2pt}
\begin{wrapfigure}{l}{25mm} 
 \includegraphics[width=1in,height=1.25in,clip,keepaspectratio]{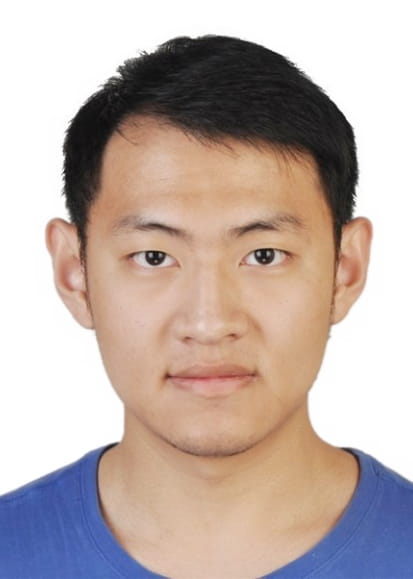}
\end{wrapfigure}\par
{\small \textbf{Shuo Liu} (Student Member, IEEE) received his B.Eng. degree in Mechanical Engineering from Chongqing University, Chongqing, China, in 2018, and his M.Sc. degree in Mechanical Engineering from Columbia University, New York, NY, USA, in 2020. He is currently a Ph.D. candidate in Mechanical Engineering at Boston University, Boston, USA and his research interests include optimization, nonlinear control, deep learning and robotics. \par}

\vspace{10pt}

\setlength{\intextsep}{2pt}
\begin{wrapfigure}{l}{25mm} 
 \includegraphics[width=1in,height=1.25in,clip,keepaspectratio]{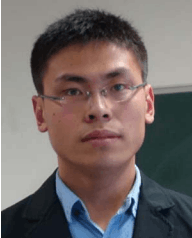}
\end{wrapfigure}\par
{\small \textbf{Wei Xiao} (Member, IEEE) received the B.Sc. degree
in mechanical engineering and automation from the University of Science and Technology Beijing, Beijing, China, the M.Sc. degree in robotics from the Chinese Academy of Sciences (Institute of Automation), Beijing, China, and the Ph.D. degree in systems
engineering from Boston University, Boston, MA, USA, in 2013, 2016, and 2021, respectively.

He is currently an Assistant Professor at Worcester Polytechnic Institute (WPI), Worcester, MA, USA. His current research interests include control theory and machine learning, with particular emphasis on robotics and traffic control. \par}

\vspace{10pt}

\setlength{\intextsep}{2pt}
\begin{wrapfigure}{l}{25mm} 
\includegraphics[width=1in,height=1.25in,clip,keepaspectratio]{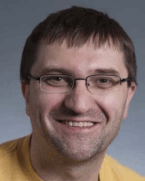}
\end{wrapfigure}\par
{\small \textbf{Calin Belta} (Fellow, IEEE) received the B.Sc.
and M.Sc. degrees from the Technical University of Iasi, Iasi, Romania, in 1995 and 1997, respectively, and the M.Sc. and Ph.D. degrees from the University of Pennsylvania, Philadelphia, PA, USA, in 2001 and 2003. 

He is currently the Brendan Iribe Endowed Professor of Electrical and Computer Engineering and Computer Science at the University of Maryland, College Park, MD, USA. His research focuses on control theory and formal methods, with particular emphasis on hybrid and cyber-physical systems, synthesis and verification, and applications in robotics and biology.

\end{document}